\pdfoutput=1
\documentclass[11pt]{article}
\usepackage{amsmath}
\usepackage{amsfonts}
\usepackage{amsthm}
\usepackage{fullpage}
\usepackage{color,soul}

\usepackage{comment}

\usepackage{url}
\usepackage{algorithm}
\usepackage[noend]{algpseudocode}
\usepackage{array}
\usepackage{xy}
\usepackage{setspace}
\usepackage{multicol}
\usepackage{algorithm}
\usepackage{hyperref}
\usepackage{bbm}
\usepackage{multirow}

\usepackage{todonotes}
\makeatletter
\if@todonotes@disabled

\else

\fi
\makeatother

\makeatletter
\if@todonotes@disabled

\else

\fi
\makeatother

\makeatletter
\def\th@plain{%
	\thm@notefont{}% same as heading font
	\itshape % body font
}
\def\th@definition{%
	\thm@notefont{}% same as heading font
	\normalfont % body font
}
\makeatother

\usepackage{hyperref}
\hypersetup{
	colorlinks = true,
	linkbordercolor = {white},
	citecolor = {blue},
	linkcolor = {blue}
}

\newcommand{\poly}{\mathrm{poly}}
\DeclareMathOperator{\dtw}{DTW}
\DeclareMathOperator{\dtep}{DTEP}

\newcommand{\wt}[1]{\widetilde{#1}}
\newcommand{\wb}[1]{\overline{#1}}

\DeclareMathOperator{\sk}{sk}
\DeclareMathOperator{\RR}{\mathbb{R}}

\DeclareMathOperator{\EE}{\mathbb{E}}

\newcommand{\overbar}[1]{\mkern 1.5mu\overline{\mkern-1.5mu#1\mkern-1.5mu}\mkern 1.5mu}

 \newcommand{\E}{\mathbb{E}}

\newcommand{\ed}{\operatorname{ed}}

\newtheorem{thm}{Theorem}[section]
\newtheorem{defn}[thm]{Definition}
\newtheorem{lem}[thm]{Lemma}
\newtheorem{prop}[thm]{Proposition}

\newtheorem{cor}[thm]{Corollary}

\newtheorem{theorem}[thm]{Theorem}
\newtheorem{definition}[thm]{Definition}
\newtheorem{lemma}[thm]{Lemma}

\theoremstyle{remark}

\title{The One-Way Communication Complexity of Dynamic Time Warping Distance}
\author{Vladimir Braverman\\Johns Hopkins University \\ \texttt{vova@cs.jhu.edu} \and Moses Charikar\\Stanford University \\ \texttt{moses@cs.stanford.edu} \and William Kuszmaul\\Stanford University \\\texttt{kuszmaul@cs.stanford.edu}\and David P. Woodruff\\ Carnegie Mellon University \\ \texttt{dwoodruf@cs.cmu.edu} \and Lin F. Yang\\Princeton University \\ \texttt{lin.yang@princeton.edu}}

\date{}
\begin{document}
\maketitle

\begin{abstract}
	We resolve the randomized one-way communication complexity of Dynamic
Time Warping (DTW) distance.  We show that there is an efficient
one-way communication protocol using $\widetilde{O}(n/\alpha)$ bits for
the problem of computing an $\alpha$-approximation for DTW between
strings $x$ and $y$ of length $n$, and we prove a lower bound of
$\Omega(n / \alpha)$ bits for the same problem. Our communication
protocol works for strings over an arbitrary metric of polynomial size
and aspect ratio, and we optimize the logarithmic factors depending on
properties of the underlying metric, such as when the points are
low-dimensional integer vectors equipped with various metrics or have
bounded doubling dimension. We also
consider linear sketches of DTW, showing that such sketches must have
size $\Omega(n)$.

%We also present several results on the time complexity of DTW. First, we give an 
%$O(n \dtw)$ time algorithm which is better than the $O(n^2)$ time algorithm when $\dtw$
%is small. We then use this to obtain an $\tilde{O}(n^2/\alpha)$-time algorithm providing a
%$\alpha$-approximation. This gives the first sub-quadratic approximation algorithm for $\dtw$. 
%Finally, we give several applications of our techniques to the edit distance. 
 
\end{abstract}
	\thispagestyle{empty}
	\newpage
	\setcounter{page}{1}

\section{Introduction}\label{sec:intro}

%\begin{itemize}
%\item 
%Introduction to DTW: 
The Dynamic Time Warping (DTW) distance is a widely used %similarity heuristic
distance measure
between time series. 
It is particularly flexible in dealing with temporal sequences that vary in speed.
To measure the distance between two sequences, each sequence is ``warped" non-linearly in the time dimension (i.e., portions of each sequence are stretched by varying amounts) and the warped sequences are compared by summing up distances between corresponding elements.
DTW was popularized in the
 speech recognition
community by Sakoe and Chiba \cite{dtwband}. It
 was introduced in
the data mining community for mining time series by 
 Berndt and
Clifford \cite{bc94}. 
 It has many applications include phone
authentication \cite{dtwapp1}, signature
verification \cite{dtwapp2},
 speech recognition \cite{dtwapp3},
bioinformatics \cite{dtwapp4},
 cardiac medicine \cite{dtwapp5}, and
song identification
\cite{dtwapp6}. Several techniques and heuristics have been developed
to speed up natural dynamic programming algorithms for it
\cite{h75,dtwband,kp99,kp00,k02,BUWK15,PFWNCK16}. We refer the reader
also to Section 2 of \cite{heuristicapproaches} for more references.

%\item 
% mention triangle inequality
Distance measures on sequences and time series have been extensively studied in the literature. %in the theory community.
%Recently there has been interest in the sketchability of strings, time series, and geometric data. 
Given two sequences $x = x_1, x_2, \ldots, x_m$ and $y = y_1, y_2, \ldots, y_n$ of points in $\mathbb{R}^d$ (or a metric space), one seeks to ``match the points up'' as closely as possible. 
One way of doing this is to 
%\hlmosesdefine a correspondence $(\bar{x}, \bar{y})$}{Shouldn't we also explain how points are matched up?}
define a ``correspondence'' $(\bar{x}, \bar{y})$ between $x, y$ by considering expansions of $x$ and $y$ to produce sequences of equal length, i.e., we duplicate each point $x_i$ some number $m_i$ times (to produce $\bar{x}$) and each point $y_j$ some $n_j$ times (to produce $\bar{y}$), so that $\sum_{i=1}^m m_i = \sum_{i=1}^n n_i$. 
%For these new equal-length sequences, 
Now, we define a vector $z$ with $z_i = d(\bar{x}_i, \bar{y}_i)$, for some underlying distance function $d$ and choose the correspondence which minimizes a certain function of $z$.
For example, minimizing $\sum z_i$ leads to the Dynamic Time Warping distance.
Minimizing $\max_i z_i$ leads to the discrete Fr\'echet distance.
The edit distance between strings can be similarly cast in this framework.
One unusual aspect of DTW (in contrast to its close cousins, edit distance and Fr\'echet distance) is that it does not satisfy the triangle inequality.

Edit distance and Fr\'echet distance have received a lot of attention in the theory community.
Fundamental questions such as exact and approximation algorithms, nearest neighbor search, sketching, and communication complexity have been intensively studied.
However, there are relatively few results about DTW.
%Whereas DTW's cousins edit distance and Frechet distance have been extensively studied in the theory community, until recently very little was known about DTW. 
Similar to edit distance, DTW can be computed by a quadratic-time dynamic program. 
Recently, it was shown that there is no strongly subquadratic-time algorithm for DTW unless the Strong Exponential Time Hypothesis is false \cite{DTWhard, DTWhard2}; approximation algorithms for DTW were obtained under certain assumptions about properties of the input strings \cite{DTWapprox1, DTWapprox2}; and slightly subquadratic algorithms for DTW have also been obtained \cite{DTWsubquadratic}. 
$\dtw$ was studied in the context of LSH \cite{ds17} and nearest neighbor search \cite{syf05,ep17}.
To the best of our knowledge, until now, there has been no study of the communication complexity of this basic distance measure on sequences. 

%\item 
In this paper, we study the {\em one-way communication complexity} of DTW.
For a distance measure $d: \mathcal{X} \times \mathcal{X} \rightarrow \mathbb{R}^{\geq 0}$ such as DTW,
the goal in the one-way communication model is to define a randomized function $S$ and an estimation procedure $E$ so that for any $x, y \in \mathcal{X}$, given $S(x)$ and $y$, the output $E(S(x), y) \approx d(x,y)$ with large probability. 
%Here $S(x)$ is referred to as the {\it sketch} of $x$. 
There are various notions of approximation, but a natural one is that $d(x,y) \leq E(S(x), y) < \alpha d(x,y)$ for an approximation factor $\alpha > 1$.
The challenge is to understand how large $S(x)$ needs to be (for sequences of length $n$) in order to obtain approximation factor $\alpha$.
A closely related notion is that of sketching, where the estimation procedure takes $S(x)$ and $S(y)$ and we require that $E(S(x), S(y)) \approx d(x,y)$ with large probability.
This one-way communication complexity question has been studied previously for edit distance, in the context of document exchange \cite{bz16, belazzougui2015efficient, irmak2005improved}.
This model captures a number of applications, e.g., lower bounds in it apply to data stream algorithms and to sketching protocols. Upper bounds in it are appropriate for nearest neighbor search; indeed, the natural thing to do here is a lookup table, so a (one-way) sketch of size $b$ bits creates a table of size $2^b$ (see e.g., \cite{ADIW}).
One-way communication 
is one of the simplest and most natural settings in which one can study communication complexity,
and it has rich connections to areas such as information theory, coding theory, on-line computing, and learning
theory \cite{KNR}.

%We show tight upper and lower bounds for the one-way communication complexity of DTW.
%We note that our lower bound of $\Omega(n / \alpha)$ applies even to
%approximating DTW in the low distance regime (i.e., distinguishing
%$\dtw(x,y) \leq 1$ versus $\dtw(x,y) > \alpha$ with constant
%probability), and that in this regime the edit distance admits a much
%smaller-sized sketch \cite{bz16, edsketch}. To the best of our knowledge, our result provides the first
%separation between the $\dtw$ and the edit distance.
%\end{itemize}   

\subsection{Our Results}
Our main result is a tight $\widetilde{\Theta}(n/\alpha)$ bound, up to
logarithmic factors, on the one-way communication complexity of
computing an $\alpha$-approximation to $\dtw$. The results are
discussed in more detail below.

We present a communication protocol using $\widetilde{\Theta}(n/\alpha)$
bits which works for $\dtw$ on any underlying metric space of
polynomial size and aspect ratio (Theorem
\ref{thmDTWsketchgeneral}). We optimize the logarithmic factors in 
the important case when the points are natural numbers and the
distance $d(a,b) = |a-b|$, as
well as more generally when the points are low-dimensional integer
vectors equipped with various metrics (Theorem \ref{thmmainthmgeneralized}); we also optimize for the
important case where the underlying metric has small doubling
dimension (Theorem \ref{thmmainthmgeneralized}). At the cost of an
extra logarithmic factor in complexity, all of our protocols are also
time-efficient, in that Alice and Bob each run in polynomial time.

Next, we turn to lower bounds. Our communication protocol is
non-linear, and we show that in general linear sketches must have size
$\Omega(n)$ (Theorem \ref{thmlinear}). Moreover, we prove that our
upper bounds are within a polylogarithmic factor of tight,
establishing a randomized one-way communication lower bound of
$\Omega(n/\alpha)$ for any underlying metric space of size at least
three, for one-way communication algorithms which succeed with
constant success probability (Theorem \ref{thmmainlow}). We optimize
this in several ways: (1) when the underlying metric is generalized
Hamming space over a point set of polynomial size $n^{1 + \Omega(1)}$,
we improve the lower bound to $\Omega(n/\alpha \cdot \log n)$ for
algorithms which succeed with probability $1-1/n$, and show this is
optimal (Theorem \ref{thmgenhamcomplexity}); (2) for the natural
numbers, we improve the lower bound to $\Omega(n/\alpha \cdot
\log(\min(\alpha, |\Sigma|)))$ for algorithms which succeed with
probability at least $1-1/\min(\alpha, |\Sigma|)$ (Theorem
\ref{thmintlower}).  We note that our lower bound of $\Omega(n /
\alpha)$ applies even to approximating DTW in the low distance regime
(i.e., distinguishing $\dtw(x,y) \leq 1$ versus $\dtw(x,y) > \alpha$
with constant probability), and that in this regime the edit distance
admits a much smaller sketching complexity \cite{bz16, edsketch}. To
the best of our knowledge, our result provides the first separation
between the $\dtw$ and the edit distance.

We summarize our results in Table~\ref{tab:res}.
\begin{table}[htb!]
	\renewcommand{\arraystretch}{1.2}%
	\centering
	\begin{tabular*}{\textwidth}{c @{\extracolsep{\fill}}cccc}
		\hline\hline
		Model & Metric Space & Communication Bounds & Theorem\\
		\hline
		%One-way & General Finite & $O(n\alpha^{-1}\cdot\log^3 n)$ &  \\
		\multirow{6}{*}{One-way} & Finite & $O(n/\alpha\cdot\log \alpha\cdot\log^3 n)$$^*$ &   \ref{thmDTWsketchgeneral}
		\\
		%& Natural Numbers & $O(\sigma n\alpha^{-1}\cdot \log \alpha\cdot\log^2 n \cdot \log\log\log n)$ & Time Efficient\\
		& Natural Numbers & $O(n/\alpha\cdot \log \alpha\cdot \log^2 n \cdot \log\log\log n)$$^*$ &  \ref{cor:414}\\
		& $\ell_p^d$ & $O_{p, d}(n/\alpha\cdot \log \alpha\cdot \log^2 n \cdot \log\log\log n)$$^*$ &\ref{cor:414}\\
		& doubling constant $\lambda$ &   $O({\log \lambda \cdot n/\alpha}\cdot \log \alpha
  \cdot\log^2 n \cdot \log \log \log n)$$^*$&\ref{cor:doubling}\\
		 & Finite&  $\Omega(n/\alpha)$& \ref{thmmainlow}\\
		 & Generalized Hamming&  $\Theta(n/\alpha\cdot \log n)$$^\dagger$& \ref{thmgenhamcomplexity}\\
		 & Natural Numbers&  $\Omega(n/\alpha\cdot \log \min(\alpha, |\Sigma|))$$^\dagger$& \ref{thmintlower}\\
		Linear Sketch & Finite&  $\Omega(n)$& \ref{thmlinear}\\
		\hline
	\end{tabular*}
	\caption{Summary of results on computing $\alpha$ multiplicative approximation of DTW$_n$ over a metric space $\Sigma$ with aspect ratio $\poly (n)$, and for the finite metric upper bounds the size of the metric is also $\poly(n)$. $^*$These upper bounds are also time efficient. Inefficient protocols can remove an additional $\log \alpha$ factor in the communication complexity. $^\dagger$These lower bounds hold for protocols that are correct with probability $1-1/n$ or $1-1/\min(\alpha, |\Sigma|)$. The doubling constant of a metric
$(\Sigma, d)$ is $\lambda$ if for all $x \in \Sigma$ and $r > 0$, the
ball $B(x, 2r)$ can be covered by $\lambda$ balls of radius $r$. \label{tab:res}}
 
\end{table}
 The layout of the paper is as follows: We present
preliminaries in Section \ref{secprelim}. We give a detailed overview
of our techniques and results in Section~\ref{sec:overview}. In
Section \ref{seccomm} we present our upper bounds on the one-way
communication complexity of DTW. Finally, in Section \ref{sec:lb}, we
presetn our lower bounds.

\section{Preliminaries}\label{secprelim}

As a convention, we say an event occurs with \emph{high probability}
if it happens with probability at least $1 - \frac{1}{\poly(n)}$ for a
polynomial of our choice. Throughout the paper, we use $(\Sigma, d)$
to denote a finite metric space. We denote by $\Sigma^n$ the set of
strings of length $n$ over $\Sigma$ and by $\Sigma^{\le n}$ the set of
strings of length at most $n$ over $\Sigma$. An important property of
$\Sigma$ will be its \emph{aspect ratio}, which is defined as the
ratio between the diameter of $\Sigma$ and the smallest distance
between distinct points in $\Sigma$.

\paragraph{Dynamic Time Warping Distance} We study the
\emph{dynamic warping distance} (DTW) of strings $x, y\in \Sigma^{\le
  n}$.  Before we formally define the DTW distance, we first introduce
the notion of an \emph{expansion} of a string.
\begin{definition}
	The \emph{runs} of a string $x\in \Sigma^{\le n}$ are the maximal
        substrings consisting of a single repeated letter. Any string
        obtained from $x$ by extending $x$'s runs is an
        \emph{expansion} of $x$.
\end{definition}
For example, the runs of $aabbbccd$ are $aa$, $bbb$, $cc$, and
$d$. Given a string $x$, we can \emph{extend} a run in $x$ by further
duplicating the letter which populates the run. For example, the
second run in $aabbbccd$ can be extended to obtain $aabbbbccd$, and we
say the latter string is an expansion of the first.

Using the notion of an expansion, we can now define dynamic time warping.
\begin{definition}
Consider two strings $x, y\in \Sigma^{\le n}$. A
\emph{correspondence}\footnote{A related concept, \emph{traversal}, is
  sometimes used in the literature. A traversal can be viewed as the
  the set of matching edges of a correspondence.}  between $x$ and $y$
is a pair $(\overline{x}, \overline{y})$ of equal-length expansions of
$x$ and $y$. The \emph{edges} in a correspondence are the pairs of
letters $(\overline{x}_i, \overline{y}_i)$, and the \emph{cost} of an
edge is given by $d(\overline{x}_i, \overline{y}_i)$. The \emph{cost}
of a correspondence is the sum $\sum_i d(\overline{x}_i,
\overline{y}_i)$ of the costs of the edges between the two
expansions. A correspondence between $x$ and $y$ is said to be
\emph{optimal} if it has the minimum attainable cost, and the
resulting cost is called the \emph{dynamic time warping distance}
$\dtw(x, y)$.
\end{definition}

When discussing a correspondence $(\overline{x}, \overline{y})$, the following terms will be useful.

\begin{definition}
A run in $\wb{x}$ \emph{overlaps} a run in $\wb{y}$ if there is an
edge between them.  A letter $x_i$ is \emph{matched} to a letter $y_j$
if the extended run containing $x_i$ overlaps the extended run
containing $y_j$.
\end{definition}

Note that any minimum-length optimal correspondence between strings
$x, y \in \Sigma^{\le n}$ will be of length at most $2n$. In
particular if in an optimal correspondence a run $r_1$ in $x$ and a
run $r_2$ in $y$ have both been extended and overlap by at least one
letter, then there is a shorter optimal correspondence in which the
length of each run is reduced by one. Thus any minimum-length optimal
correspondence has the property that every edge $(\overline{x}_i,
\overline{y}_i)$ contains at least one letter from a run that has not
been extended, thereby limiting the length of the correspondence to at
most $2n$.

DTW can be defined over an arbitrary metric space $(\Sigma, d)$, and
is also well-defined when $d$ is a distance function not satisfying
the triangle inequality.

Throughout our proofs, we will often refer to DTW over generalized
Hamming space, denoted by $\dtw_0(x, y)$. As a convention, regardless
of what metric space the strings $x$ and $y$ are initially taken over,
$\dtw_0(x, y)$ is defined to be the DTW-distance between $x$ and $y$
obtained by redefining the distance function $d(\cdot, \cdot)$ to
return $1$ on distinct inputs.

\paragraph{One-Way Communication Complexity}

In this paper, we focus on the one-way communication model. In this
model, Alice is given an input $x$, Bob is given an input $y$, and Bob
wishes to recover a valid solution to a problem with some solution-set
$f(x, y) \subseteq \mathbb{R}$. (For convenience, we will refer to the
problem by its solution set $f(x, y)$.) Alice is permitted to send Bob
a single message $\sk(x)$, which may be computed in a randomized
fashion using arbitrarily many public random bits. Bob must then use
Alice's message $\sk(x)$ in order to compute some $F(\sk(x), y)$,
which he returns as his proposed solution to $f(x, y)$.

The pair $(\sk, F)$ is a \emph{$p$-accurate one-way
  communication protocol} for the problem $f(\cdot, \cdot)$ if for all
$x$ and $y$, the probability $\Pr[F(\sk(x), y) \in f(x, y)]$ that Bob
returns a correct answer to $f(x, y)$ is at least $p$. The protocol is
said to have \emph{bit complexity} at most $m$ if Alice's message
$\sk(x)$ is guaranteed not to exceed $m$ in length. Moreover, the
protocol is said to be \emph{efficient} if both $\sk$ and $F$ can be
evaluated in time polynomial in the length of $x$ and $y$.

Fix a parameter $p \in (0,1]$, the randomized \emph{one-way
    communication complexity} $\mathrm{CC}_{p}(f)$ of the problem $f$
  is the minimum attainable bit complexity of a $p$-accurate one-way
  communication protocol for $f$. %%  That is, 
	%% \[
	%% 	\mathrm{CC}_{p}(f) = \min_{\Pi:=(\sk, F)}\max_{x}|\sk(x)|,
	%% \]
	%% where $\Pi$ ranges over $p$-accurate one-way protocols for $f$.
The focus of this paper is on the one-way communication complexity of
the $\alpha$-DTW problem, defined as follows:
\begin{definition}[$\alpha$-DTW]
  The $\alpha\text{-DTW}(\Sigma^{\le n})$ problem is parameterized
  by an approximation parameter $1 \le \alpha \le n$. The inputs are a
  string $x \in \Sigma^{\le n}$ and a string $y \in \Sigma^{\le
    n}$. The goal is recover an $\alpha$-approximation for $\dtw(x,
  y)$. In particular, the set of valid solutions is
  \[\{t \mid \dtw(x, y) \le t < \alpha \cdot \dtw(x, y)\}.\]
\end{definition}

One can also consider the decision version of this problem, in which
one wishes to distinguish between distances at most $r$ and distances
at greater than $r\alpha$:
\begin{definition}[DTEP]
  The Decision Threshold Estimation Problem
  $\dtep_{r}^\alpha(\Sigma^{\le n})$, is paramaterized by a positive
  threshold $r > 0$ and an approximation parameter $1 \le \alpha \le n$. The
  inputs to the problem are a string $x \in \Sigma^{\le n}$ and a
  string $y \in \Sigma^{\le n}$. An output of $0$ is a valid solution
  if $\dtw(x, y) \le r\alpha$, and an output of $1$ is a valid
  solution of $\dtw(x, y) > r$.
\end{definition}

Notice that any algorithm for $\alpha$-DTW immediately gives an
solution for $\dtep_{r}^{\alpha}$ for any $r>0$. Conversely, any lower
bound for the communication complexity of $\dtep$ gives a lower bound
for the communication complexity of $\alpha$-DTW. For both of the
above two definitions, we may omit the sequence space $\Sigma^{\le n}$
if it is clear from the context.

\section{Technical Overview}\label{sec:overview}

In this section, we present the statements and proof overviews of our
main results.%%  The full proofs are presented in the extended version of
%% the paper, which we include as an appendix.

\paragraph{Complexity Upper Bounds:}

Our starting point is the following: suppose that $x, y \in \Sigma^n$
for a metric space $\Sigma$ of polynomial size and aspect ratio, and
further that the distances between points are always either $0$ or at
least $1$. Alice and Bob wish to construct a $2/3$-accurate one-way
protocol for $\alpha$-DTW.  \\\\ {\it Collapsing Repeated Points.}
Consider the strings $c(x)$ and $c(y)$, formed by reducing each run of
length greater than one in $x$ and $y$ to the same run of length
one. If we define $l$ to be the length of the longest run in $x$ or
$y$, then $\dtw(x, y) \leq l \cdot \dtw(c(x),c(y))$. Indeed, any
correspondence $(\overbar{c(x)}, \overbar{c(y)})$ between $c(x)$ and
$c(y)$ gives rise to a correspondence $(\bar{x}, \bar{y})$ between $x$
and $y$ obtained by duplicating each coordinate in $\overbar{c(x)}$
and $\overbar{c(y)}$ a total of $l$ times. Moreover, since any
correspondence $(\bar{x}, \bar{y})$ between $x$ and $y$ is also a
correspondence between $c(x)$ and $c(y)$, it follows that $\dtw(c(x),
c(y)) \leq \dtw(x,y)$.  \\\\ {\it Inefficient Protocol via Hashing.}
Suppose Alice and Bob are guaranteed that $\dtw(x, y) \leq n/\alpha$,
and that the maximum run-length $l$ satisfies $l < \alpha$. Then it
suffices for Alice and Bob to compute $\dtw(c(x), c(y))$; and for this
it suffices for Bob to be able to reconstruct $c(x)$. The claim is
that from a random hash of $c(x)$ of length $O(n/\alpha \log n)$ bits,
given $c(y)$, Bob can reconstruct $c(x)$. Indeed, given that
$\dtw(c(x), c(y)) \le n / \alpha$, and given that the runs in $c(x)$
and $c(y)$ are all of length one, one can verify that there must be an
optimal correspondence $(\wb{c(x)}, \wb{c(y)})$ between $c(x)$ and
$c(y)$ such that $\wb{c(y)}$ is obtained from $c(y)$ by extending at
most $n / \alpha$ runs. Since there are $n^{O(n / \alpha)}$ ways to
choose which runs in $c(y)$ are extended, and since there are then
$n^{O(n / \alpha)}$ ways to choose the new lengths to which those runs
are extended, it follows that there are only $n^{O(n / \alpha)}$
options for $\wb{c(y)}$. Moreover, because $\overline{c(x)}$ and
$\overline{c(y)}$ differ in at most $n / \alpha$ positions, for a
given option of $\overline{c(y)}$ there are only $n^{O(n / \alpha)}
\cdot |\Sigma|^{O(n / \alpha)} = n^{O(n / \alpha)}$ options for
$\wb{c(x)}$ and thus for $c(x)$. Since starting from $c(y)$, there are
only $n^{O(n / \alpha)}$ options for $c(x)$, meaning that a $O(n /
\alpha \log n)$-bit hash allows Bob to recover $c(x)$ with high
probability.  \\\\ {\it Efficiency via Edit Distance Sketch.}  In
addition to requiring that $\dtw(x, y) \le n / \alpha$ and $l <
\alpha$, the above protocol is inefficient since Bob needs to
enumerate over all possibilities of $c(x)$ and compute the hash value
of each. Exploiting the fact that $c(x)$ and $c(y)$ contain only runs
of length one, we prove that $\dtw(c(x),c(y))$ is within a constant
factor of the edit distance between $c(x)$ and $c(y)$. This means that
Alice can instead invoke the edit-distance communication protocol of
\cite{edsketch} of size $O(n/\alpha \log n \log \alpha)$, which allows
Bob to efficiently recover $c(x)$ using the fact that the edit
distance between $c(x)$ and $c(y)$ is $O(n/\alpha)$.  \\\\ {\it
  Handling Heavy Hitters.}  The arguments presented so far require
that $x$ and $y$ contain no runs of length greater than $\alpha$. We
call such runs \emph{heavy hitters}. To remove this restriction, a key
observation is that there can be at most $n/\alpha$ heavy
hitters. Therefore Alice can communicate to Bob precisely which runs
are heavy hitters in $x$ using $O(n/\alpha \log n)$ bits. The players
then proceed as before: Alice collapses her input $x$ to $c(x)$ by
removing consecutive duplicates, and Bob collapses his input $y$ to
$c(y)$ by removing consecutive duplicates.  We still have $\dtw(c(x),
c(y)) \leq \dtw(x,y)$ since any correspondence between $x$ and $y$ is
a correspondence between $c(x)$ and $c(y)$. Thus, as before, Bob can
reconstruct $c(x)$ whenever $\dtw(x,y) \leq n/\alpha$. Now, though, it
could be that $\dtw(x,y) > \alpha \dtw(c(x), c(y))$ because of the
positions in $c(x)$ and $c(y)$ that occur more than $\alpha$
times. However, Bob uses his knowledge of the locations and values of
the heavy hitters, together with $c(x)$, to create a string $x'$
formed from $x$ by collapsing runs of length less than $\alpha$, and
not doing anything to runs of length at least $\alpha$.  Now by
computing $\dtw(x', y)$, Bob obtains a $\alpha$-approximation for
$\dtw(x, y)$, since any correspondence between $x'$ and $y$ gives rise
to a correspondence between $x$ and $y$ by duplicating each letter
$\alpha$ times.

Having handled the heavy hitters, the only remaining requirement by
our protocol is that the distances between letters in $x$ and $y$ be
zero and one. Thus we arrive at the following:
\begin{prop}[Protocol over Hamming Space]
 Consider $\dtw$ over a metric space $\Sigma$ of polynomial size with
 distances zero and one. Then for $p = 1 - {\poly(n^{-1})}$, there is
 an efficient $p$-accurate one-way communication protocol for
 $\alpha$-DTW over $\Sigma^{\le n}$ which uses $O\left({n}\alpha^{-1}
 \cdot\log \alpha \cdot \log n\right)$ bits.  Moreover, for any
 $\delta \in (0, 1)$, there is an inefficient $(1-\delta)$-accurate
 protocol for $\alpha$-DTW$(\Sigma^{\le n})$ using space
 $O({n}\alpha^{-1} \cdot \log n + \log\delta^{-1})$ for any
 $\delta\in(0,1)$.
 \label{propprotocolhamming}
\end{prop}

Note that our protocol is constructive in that it actually allows for
$y$ to build a correspondence between $x$ and $y$ satisfying the
desired approximation bounds.

In generalizing to DTW over arbitrary metric spaces, we will use our
protocol over Hamming Space as a primitive. Moreover, we will exploit
the fact that it can be used to solve a slightly more sophisticated
problem which we call \emph{bounded $\alpha$-DTW}:
\begin{definition}[Bounded $\alpha$-DTW]
	\label{def:bounded-alpha-dtw}
In the bounded $\alpha$-DTW$(\Sigma^{\le n})$ problem, Alice and Bob
are given strings $x$ and $y$ in $\Sigma^{\le n}$. The goal for Bob
is:
\begin{itemize}
	\item If $\dtw_0(x, y) \le n / \alpha$, solve $\alpha$-$\dtw$ on $(x, y)$.
	\item If $\dtw_0(x, y) > n / \alpha$, either solve
          $\alpha$-$\dtw$ on $(x, y)$, or return ``Fail''.
\end{itemize}
\end{definition}
A crucial observation is that Proposition \ref{propprotocolhamming}
continues to hold without modification if the alphabet $\Sigma$ has
arbitrary distances and our goal is to solve the bounded $\alpha$-DTW
problem.
\\\\ {\it Extending Distance Range via HSTs.}  The result for the
bounded $\alpha$-DTW problem allows for Bob to either determine an
$\alpha$-approximation for $\dtw(x, y)$, or to determine that $\dtw(x,
y) > n / \alpha$. As a result the algorithm can be used to distinguish
between $\dtw(x, y) \le n / \alpha$ an $\dtw(x, y) > n$. One issue
though is that the argument cannot distinguish between larger
distances, such as for example between the cases $\dtw(x,y) \leq n$
and $\dtw(x,y) > n \alpha$. A key idea for resolving this issue is to
first consider the $\dtw$ problem over a $2$-hierarchically
well-separated tree metric (HST), and then use the embedding of
\cite{trees} to embed an arbitrary finite metric of polynomial size
and aspect ratio into such a metric. A $2$-hierarchically
well-separated tree metric is defined as the shortest path metric on a
tree whose nodes are elements of $\Sigma$ and whose edges have
positive weights for which on any root-to-leaf path, the weights are
geometrically decreasing by a factor of $2$. Since the weights
decrease geometrically, for convenience we define pairwise distances
in the tree metric to be the maximum edge length on the tree path
between the nodes, a notion of distance which coincides with the sum
of edge lengths up to a constant factor.

Suppose the points in $\Sigma$ correspond to a $2$-hierarchically
well-separated tree metric and we wish to distinguish between whether
$\dtw(x, y) \le nr / \alpha$ or $\dtw(x, y) > nr$. A crucial idea is
what we call the $r$-simplification $s_r(x)$ of a string $x$, which
replaces each character $p_i$ in $x$ with its highest ancestor in the
tree reachable via edges of weight at most $r/4$. A key property is
that $\dtw(s_r(x), s_r(y)) \leq \dtw(x,y)$, since for two points
$\ell_1, \ell_2$ in $x, y$, respectively, either they each get
replaced with the same point in the $r$-simplifications of $x$ and
$y$, or the maximu-length edge on a path between $\ell_1$ and $\ell_2$
is the same before and after $r$-simplification. Notice that if a
point in $s_r(x)$ is not equal to a point in $s_r(y)$, then their
distance is at least $r/4$, by the definition of an
$r$-simplification. Combining the preceding two observations, if
$\dtw(x,y) \leq nr/\alpha$, then $\dtw(s_r(x), s_r(y)) \le nr /
\alpha$ and there is a correspondence for which $s_r(x)$ and $s_r(y)$
disagree in at most $4n/\alpha$ positions. On the other hand, since we
only ``collapse'' edges of weight at most $r/4$, we have that if
$\dtw(x,y) > nr$, then $\dtw(s_r(x), s_r(y)) > nr/2$, since the
optimal correspondence has length at most $2n$.

It follows that the cases of $\dtw(x, y) \le nr / \alpha$ and $\dtw(x,
y) > nr$, correspond with the cases of $\dtw(s_r(x), s_r(y)) \leq
nr/\alpha$ and $\dtw(s_r(x), s_r(y)) > nr/2$, and moreover that when
$\dtw(s_r(x), s_r(y)) \leq nr/\alpha$, there is an optimal
correspondence for which $s_r(x)$ and $s_r(y)$ disagree in at most
$4n/\alpha$ positions. Thus we can use our protocol for the
$\alpha$-bounded DTW problem to figure out which case we are in, for a
given $r$. This gives a protocol for distinguishing between whether
$\dtw(x, y) \leq nr/\alpha$ or $\dtw(x, y) > nr$.

In order to obtain an $\alpha$-approximation for $\dtw(x, y)$, the
rough idea now is to run the above protocol multiple times in
parallel as $r$ varies in powers of $2$, and then to find the smallest
value of $r$ for which the protocol declares $\dtw(x,y) \leq nr$. 
This works as long as points are taken from a $2$-hierarchically 
well-separated tree metric. In order to extend the result to hold over
arbitrary finite metrics of polynomial size and aspect ratio, the
final piece is the embedding $\phi$ of \cite{trees}, which embeds any
polynomial size metric $\Sigma$ into a $2$-hierarchically
well-separated tree metric for which  for all $a,b \in \Sigma$,
$d(a,b) \leq d(\phi(a), \phi(b))$ and  ${\bf E}(d(\phi(a), \phi(b)))
= O(\log n) d(a,b)$. This ``lopsided'' guarantee is sufficient for
us since it ensures in any correspondence the sum of distances 
after performing the embedding will not shrink, while for a single
fixed optimal correspondence, by a Markov bound the sum of distances
after performing the embedding will not increase by more than an
$O(\log n)$ factor with constant probability. Putting the pieces
together we are able to obtain an efficient $2/3$-accurate one-way
communication protocol for $\alpha$-DTW using $O(n / \alpha \log \alpha
\log^3 n)$ bits. Formally, we arrive at the following theorem:

\begin{thm}[Main Upper Bound]
  Let $\Sigma$ be a metric space
  of size and aspect ratio polynomial in $n$.  Then there is an efficient
  $2/3$-accurate one-way communication protocol for $\alpha$-DTW over
  $\Sigma$ with space complexity $O\left({n}{\alpha^{-1}}\cdot \log
  \alpha \cdot\log^3 n\right)$ and an inefficient $2/3$-accurate
  one-way protocol with complexity
  $O\big({n}{\alpha^{-1}}\cdot \log^3 n)$.
  \label{thm:overviewmain}
\end{thm}

\noindent {\it Optimizing in the Case of Natural Numbers.}
We can further optimize the logarithmic factors in our upper bound when the underlying alphabet $\Sigma$ is,
for example, the natural numbers and $d(a,b) = |a-b|$. We handle the
case  $\dtw(x,y) \leq n/\alpha$ as before. However, for larger
values of $\dtw(x,y)$, we take a different approach.

We first explain the case of distinguishing $\dtw(x,y) \leq n$ versus
$\dtw(x,y) > \alpha n$. The idea is to impose a randomly shifted
grid of side length $\alpha / 4$, and to round each point in $x$
and $y$ down to the nearest smaller grid point, resulting in strings
$x'$ and $y'$.  Define a \emph{short edge} in a correspondence to be an
edge of cost at most $\alpha / 4$, and otherwise call the edge a
\emph{long edge}. We assume w.l.o.g. that any
correspondence has length at most $2n$.

Suppose first $\dtw(x,y) \leq n$, and consider an optimal
correspondence.  We will show that the effect of rounding is such
that with probability at least $2/3$, $\dtw(x', y') \le O(n)$. First
we consider what effect rounding has on the short edges.  The expected
number of short edges with endpoints that get rounded to different
grid points is at most

\begin{equation*}
  \sum_{\textrm{short edge length }l}
  \frac{l}{\alpha / 4} \leq \frac{4\dtw(x,y)}{\alpha}.
\end{equation*}
Each such edge has its length increased by at most
$\alpha / 4$ after rounding, and so the expected contribution of
short edges to the correspondence after rounding is at most
$O(\dtw(x,y))$. Since each long edge has its length increase by at most
an additive $\alpha / 4$, and its original length is at least
$\alpha/4$, its contribution changes by at most a constant factor, so
the total contribution of long edges after rounding is
$O(\dtw(x,y))$. Hence, when $\dtw(x,y) \leq n$,  with probability at
least $2/3$ after rounding, we have $\dtw(x' , y') = O(n)$.

Next suppose $\dtw(x,y) > n\alpha$, and consider any
correspondence. The total change in the cost of the correspondence
that can result from the rounding procedure is at most $2n \cdot
\alpha / 4$, since there are at most $2n$ edges in
total. Consequently the effect of rounding is such that $\dtw(x',y') > n \alpha / 2$.

It follows that when comparing the cases of $\dtw(x, y) \le n$ and
$\dtw(x, y) > n\alpha$, there is an $\Omega(\alpha)$-factor gap
between $\dtw(x', y')$ in the two cases. Further, after rounding to
grid points, all non-equal points have distance at least $\alpha/4$,
and so if $\dtw(x',y') \leq n$, then there is a correspondence on
which they differ in at most $O(n/\alpha)$ positions. Thus our
protocol for bounded $\alpha$-DTW can be applied to distinguish
between the two cases. A similar approach can be used to distinguish
between $\dtw(x, y) \le rn / \alpha$ and $\dtw(x, y) > rn$ in general,
and this can then be used to solve $\alpha$-DTW similarly as for
$2$-hierarchically well-separated tree metrics above. We save roughly
a $\log n$ factor here because we do not incur the $\log n$ factor
distortion of embedding an arbitrary metric into a tree metric.

We remark that our algorithm in the $1$-dimensional natural number
case uses a similar grid snapping as used in \cite{ds17} for their
nearest neighbor search algorithm for Frech\'et distance.  Recently,
Bringmann (personal communication) obtained a sketch for Frech\'et
distance which builds upon the ideas in
\cite{ds17} and uses $O(n/\alpha)$ bits. To the best of our knowledge,
these techniques do not yield nontrivial results for Dynamic Time
Warping, however.
\\\\ {\it A Unified Approach.} To unify the argument for
2-hierarchically well-separated tree metrics and the natural numbers,
we recall the definition of a $\sigma$-separable metric space. A
$\delta$-bounded partition of a metric space $(\Sigma, d)$ is a partition such
that the diameter of each part is at most $\delta$. A distribution
over partitions is then called $\sigma$-separating if for all $x, y
\in \Sigma$, the probability that $x$ and $y$ occur in different parts
of the partition is at most $\sigma \cdot d(x,y)/\delta$. We say
$\Sigma$ is $\sigma$-separable if for every $\delta > 0$, there exists
a $\sigma$-separating probability distribution over $\delta$-bounded
partitions of $\Sigma$. One can also define an efficient notion of
this, whereby the distribution over partitions is efficiently sampleable.

By adapting our argument for the natural numbers to $\sigma$-separable
metrics of polynomial size and aspect ratio, we obtain an efficient
$2/3$-accurate protocol for $\alpha$-$\dtw$ with bit complexity
$O(\sigma n/\alpha \log^3 n \log \log \log n)$, where the $\log \log
\log n$ comes from minor technical subtitles. For general metrics, it
is known that $\sigma = O(\log n)$, while for the natural numbers,
$\sigma = O(1)$. Consequently, our result for $\sigma$-separable
metrics captures both the result obtained using HSTs (up to a factor
of $\log \log \log n$) as well as the optimization for the natural
numbers. Moreover, the theorem allows for space savings over many
additional metrics, such as low-dimensional integer vectors equipped
with $\ell_p$-norms, metrics with bounded doubling dimension, etc.,
all of which have $\sigma \ll O(\log n)$, allowing for improvement
over our result based on HSTs. The general result we arrive at is
captured formally in the following theorem:

\begin{thm}[Extended Main Upper Bound]
   Let $(\Sigma, d)$ be a metric space of size and aspect ratio
   $\poly(n)$. Suppose that $(\Sigma, d)$ is efficiently
   $\sigma$-separable for some $1 \le \sigma \le O(\log n)$. Then
   there is an efficient $2/3$-accurate one-way communication protocol
   for $\alpha$-DTW$(\Sigma^{\le n})$ with space complexity
   $O\left({\sigma n}{\alpha^{-1}}\cdot \log \alpha \cdot\log^2
   n \cdot \log \log \log n\right)$ and an inefficient $2/3$-accurate one-way protocol with
   space complexity $O\big({\sigma
     n}{\alpha^{-1}}\cdot\log^2 n \cdot \log \log \log n\big).$
   \label{thm:overviewextended}
\end{thm}

The proof closely follows that for the natural numbers, where instead
of our randomly shifted grid, we use a random $\delta$-bounded
partition. If we are trying to distinguish $\dtw(x,y) \leq nr/\alpha$
versus $\dtw(x,y) > nr$, then we set $\delta = \Theta(r)$. Just like
for the grid, where we ``snapped'' points to their nearest grid point,
we now snap points to a representative point in each part of the
partition, obtaining two new sequences $\tilde{x}$ and $\tilde{y}$.
By using shared randomness, the representative in each part can be
agreed upon without any communication. Just like in the grid case, we
show that if $\dtw(x,y) \le nr/\alpha$, then for the optimal
correspondence, in expectation its cost increases only by a constant
factor after snapping. On the other hand, if $\dtw(x,y) > nr$, then we
show that for every correspondence, its cost decreases only by a
constant factor. The key difference is that now the
expected number of short edges with endpoints occurring in different
parts of the partition is at most $\frac{\sigma \cdot
  \dtw(x,y)}{\delta}$.

\paragraph{Complexity Lower Bounds:} The simplest of our lower bounds
comes from a reduction from a randomized $1$-way communication lower
bound for indexing over large alphabets \cite{jw13}. In this problem,
Alice is given a string $s$ in $\mathcal{U}^r$ for some universe
$\mathcal{U}$ and length parameter $r$, and Bob is given a character
$a \in \mathcal{U}$ and an index $j \in [r]$.  The goal is for Bob to
decide if $s_j = a$ with probability at least $1-1/|\mathcal{U}|$. It
is known if Alice sends a single message to Bob, then there is an
$\Omega(r \log_2 |\mathcal{U}|)$ lower bound. By reducing this
large-alphabet indexing problem to $\alpha$-DTW when $r = n /
\alpha$. To perform the reduction, Alice's input string $s = s_1,
\ldots, s_{n / \alpha} \in \mathcal{U}^{n / \alpha}$ is mapped to the
string $x = (s_1, 1), (s_2, 2), \ldots, (s_{n / \alpha}, n /
\alpha)$. Bob's inputs of $a \in \mathcal{U}$ and $j \in [r]$ are
mapped to an input string $y = (a, j), (a, j), \ldots$ in which the
character $(a, j)$ is repeated $n$ times. If $s_j = a$, then
$\dtw(x,y) = n/\alpha-1$ (due to the $n / \alpha - 1$ characters of
$x$ that do not get matched with an equal-value letter); otherwise
$\dtw(x,y) \geq n$ (due to the fact that none of the letters in $y$
can be correctly matched). This gives a reduction to $\alpha$-DTW as
desired. Using this we have an $\Omega(n/\alpha \cdot \log n)$ lower
bound for $(1 - 1/n)$-accurate $\alpha$-DTW, provided the alphabet
size $|\Sigma|$ is say, at least $n^2$. Thus we have the following theorem:
\begin{thm}[Tight Bound Over Hamming Space]
  
	Consider $1 \le \alpha \le n$, and consider the generalized
        Hamming distance over a point-set $\Sigma$ with $\Sigma$ of
        polynomial size $n^{1+\Omega(1)}$.  For $p \ge 1 -
        1/|\Sigma|^{-1}$, $\mathrm{CC}_p(\alpha\text{-DTW}(\Sigma^{\le
          n}) = \Theta[n\alpha^{-1}\cdot\log n]$.
        \label{thmhammingfirst}
\end{thm}

In order to obtain a nearly tight lower bound for arbitrary finite
metric spaces, we construct a more intricate lower bound of $\Omega(n
/ \alpha)$ which holds whenever $|\Sigma| \ge 3$. For convenience, we
describe the argument for the case of $\Sigma = \{0, 1, 2\}$
below. The lower bound is achieved via a reduction from the Index
problem in which Alice has $s \in \{0,1\}^t$, Bob has an $i \in [t]$,
and Bob would like to output $s_i$ with probability at least
$2/3$. The randomized $1$-way communication complexity of this problem
is $\Omega(t)$. We instantiate $t = \Theta(n/\alpha)$.  For each
$s_j$, if $s_j = 1$, Alice creates a string $Z(1)$ of length $3\alpha$
consisting of $\alpha$ $0$s, followed by $\alpha$ $1$s, followed by
$\alpha$ $2$s; and if $s_j = 0$, Alice creates a string $Z(0)$ of
length $2\alpha+1$ consisting of $\alpha$ $0$s, followed by a single
$1$, followed by $\alpha$ $2$s. She then concatenates $Z(s_1), Z(s_2),
\ldots, Z(s_t)$ into a single string $x$ of length $n$. Bob, who is
given an index $i \in [t]$, creates the string $y =
(012)^{i-1}(02)(012)^{t-i}$; that is, we have the length-$3$ string
$012$ repeated $i-1$ times, then the string $02$, followed by the
string $012$ repeated $t-i$ times. (We call each piece of the form
$(012)$ and $(01)$ a \emph{block}.) Notice that if $s_i = 0$, then
$\dtw(x,y) = 1$, since the single $1$ in $Z(s_i)$ can match to either
the $0$ or $2$ in the $(02)$ block of Bob's string $y$. On the other
hand, if $s_i = 1$, the entire run of $\alpha$ $1$s in $Z(s_i)$ has to
appear somewhere in the correspondence and cannot match to the $0$ or
$2$ in the $i$-th piece of Bob's string, without incurring a cost of
$\alpha$. So these $\alpha$ $1$s must either ``travel'' to blocks $j >
i$ in $y$ or blocks $j < i$ in $y$. Suppose, without loss of
generality, most of these $\alpha$ $1$s are matched to a block $j >
i$. This has a ripple effect, since it causes the $\alpha$ $2$s in the
$i$-th block to also have to travel to a block $j > i$. While this is
possible, it then means the $\alpha$ $0$s in the $(i+1)$-st block must
travel to a block even larger than $j$, etc. Eventually, we run out of
blocks to match the elements in Alice's string to since there are
$t-i$ blocks in her string that need to be matched to fewer than $t-i$
blocks in Bob's string. This ultimately forces $\dtw(x,y) \geq
\alpha$, completing the reduction from the Index problem to
$\alpha$-DTW.  The extension of this argument to arbitrary $\Sigma$
establishes that our upper bound for general metric spaces is optimal
up to a polylogarithmic factor:

\begin{theorem}[General Lower Bound]
	Let $\Sigma = \{a, b, c\}$ be three letters with a two-point
        distance function $d:\Sigma\times \Sigma \rightarrow \RR_{+}$,
        not necessarily satisfying the triangle inequality. Consider
        $1 \le \alpha \le n$.  Then
        CC$_{0.1}\big[\dtep_{r}^{\alpha}(\Sigma^{\le n})\big] =
        \Omega(n/\alpha)$.
        \label{thm:overviewgenerallower}
\end{theorem}

Our communication protocols are non-linear, and we conclude our lower
bounds by showing that linear sketches must have size
$\Omega(n)$.

\begin{theorem}[Linear Sketching Lower Bound]
	Consider $1 \le \alpha \le n$.  Then any $0.1$-error
        linear sketch for $\alpha$-$\dtw$ on $\{0,1, 2\}^{4n}$ has
        space complexity $\Omega(n)$.
        \label{thm:overviewlinear}
\end{theorem}

\section{Upper Bounds}\label{seccomm}

In this section, we present a near optimal one-way communication
protocol for $\alpha$-DTW over an arbitrary finite metric space
$\Sigma$ with the constraint that $\Sigma$ has at most polynomial size
and aspect ratio, i.e., the ratio between the largest and the
smallest distances between distinct points.

We begin in Subsection \ref{subsecbounded} by considering an easier
problem known as bounded $\alpha$-DTW, in which Bob is only required
to compute an $\alpha$-approximation for $\dtw(x, y)$ when $\dtw_0(x,
y) \le n / \alpha$, and he may instead return ``Fail'' when $\dtw_0(x,
y) > n / \alpha$ (recall that $\dtw_0$ is the number of conflicting
edges in a correspondence). Proposition \ref{propboundedKDTW} gives an
efficient one-way protocol for bounded $\alpha$-DTW.

Building on Proposition \ref{propboundedKDTW}, in Subsection \ref{subsectree}
we then design an efficient one-way protocol for $\alpha$-DTW over
well-separated tree metrics (Lemma
\ref{thmDTWsketchspecializedtree}).

Then, in Subsection \ref{subsectreetogeneral}, Theorem
\ref{thmDTWsketchgeneral} provides an efficient one-way protocol for
arbitrary finite metric spaces with polynomially bounded size and
aspect ratios by first embedding the metric space into a
well-separated tree metric and then using Lemma
\ref{thmDTWsketchspecializedtree}.

The protocol given by Theorem \ref{thmDTWsketchgeneral} uses
$\tilde{O}(n / \alpha)$ bits, which we will later see is within a
polylogarithmic factor of optimal. In Subsection
\ref{subsecgeneralgeneral}, we prove a further generalization of
Theorem \ref{thmDTWsketchgeneral} which for allows for a tighter upper
bound by a logarithmic factor for certain important cases of $\Sigma$
such as when $\Sigma \subseteq \mathbb{R}$.

\subsection{The Bounded $\alpha$-DTW Problem}\label{subsecbounded}

We define \emph{bounded $\alpha$-DTW} over a metric space $\Sigma$ to
be the following communication problem.
\begin{definition}[Bounded $\alpha$-DTW]
	\label{def:bounded-alpha-dtw}
In the bounded $\alpha$-DTW$(\Sigma^{\le n})$ problem,
Alice and Bob are given strings $x$
and $y$ in $\Sigma^{\le n}$. The goal for Bob is:
\begin{itemize}
	\item If $\dtw_0(x, y) \le n / \alpha$, solve $\alpha$-$\dtw$ on $(x, y)$.
	\item If $\dtw_0(x, y) > n / \alpha$, either solve
          $\alpha$-$\dtw$ on $(x, y)$, or return ``Fail''.
\end{itemize}

\end{definition}

In order to design an efficient one-way communication scheme for
bounded $\alpha$-DTW, we will use what we refer to as the
\emph{$K$-document exchange} problem as a primitive. Here, Alice and
Bob are given strings $x$ and $y \in \Sigma^n$. The goal for Bob is:
\begin{itemize}
\item If $\ed(x, y) \le K$, recover the string $x$.
\item If $\ed(x, y) > K$, either recover $x$ or return ``Fail''.
\end{itemize}

The $K$-document exchange problem has been studied extensively
\cite{orlitsky1991interactive, jowhari2012efficient,
  belazzougui2015efficient, chakraborty2015low, edsketch}. The one-way
communication protocol of \cite{edsketch} efficiently solves
$K$-document exchange using $O(K \log (n/K) \cdot \log n)$ bits with
high probability. This can be slightly improved at the cost of being
no longer time-efficient using the protocol of
\cite{orlitsky1991interactive}, which achieves accuracy $1 - \delta$
for any $\delta \in (0, 1)$ by having Alice simply hash her string to
a $\Theta(K \cdot\log n + \log \delta^{-1})$-bits.

The $K$-document exchange problem concerns edit distance rather than
DTW. Nonetheless, in designing a sketch for DTW, the $K$-document
exchange problem will prove useful due to a convenient relationship
between edit distance and DTW over generalized Hamming space (or equivalently $\dtw_0$).

\begin{lem}[DTW$_0$ Approx. Edit Dist.]
	\label{lem2approx}
  Let $x, y$ be strings of length at most $n$ with letters from any
  metric space, and suppose that neither string contains any runs of
  length greater than one. Then
  $$\dtw_0(x, y) \le \ed(x, y) \le 3 \dtw_0(x, y).\footnote{This can be improved to $2 \dtw_0(x, y)$ with slightly better bookkeeping.}$$
\end{lem}
\begin{proof}
  We first show that $\dtw_0(x, y) \le \ed(x, y)$. A sequence of edits
  between $x$ and $y$ can be thought of as a series of insertions in
  each of $x$ and $y$, as well as substitutions. One can create
  expansions $\overline{x}$ and $\overline{y}$ of $x$ and $y$,
  respectively, by extending runs by one in each place where the
  sequence of edits would have performed an insertion. The Hamming
  distance between $\overline{x}$ and $\overline{y}$ is then at most
  the length of the sequence of edits. Hence $\dtw_0(x, y) \le \ed(x,
  y)$.

  Next we show that $\ed(x, y) \le 3 \dtw_0(x, y)$. Consider an
  optimal correspondence $(\overline{x}, \overline{y})$ between $x$ and
  $y$. Without loss of generality, we may assume that
  whenever two runs in $\overline{x}$ and $\overline{y}$ overlap, at
  least one of them has length only one. (Indeed, otherwise both runs
  could have been reduced in size by one at no cost to DTW.)
  Therefore, any run of length $k$ in $\overline{x}$ must overlap $k$
  distinct runs in $\overline{y}$, and thus must incur at least $(k -
  1) / 2$ Hamming differences. On the other hand, because the run is
  length $k$, the expansion of the run can be simulated by $k - 1$
  insertions. Therefore, $\overline{x}$ and $\overline{y}$ can be
  constructed from $x$ and $y$ through at most $2\dtw_0(x, y)$
  edits. Hence, $\ed(x, y) \le 3 \dtw_0(x, y)$.
\end{proof}

We now present an efficient one-way communication scheme for bounded
$\alpha$-DTW. (Note that this also implies Proposition
\ref{propprotocolhamming} from Section \ref{sec:overview}.)

\begin{prop}[Protocol for Bounded DTW]
	
 Consider $\dtw$ over a metric space $\Sigma$ of polynomial size. Then
 for $p = 1 - {\poly(n^{-1})}$, there is an efficient $p$-accurate
 one-way communication protocol for bounded $\alpha$-DTW over
 $\Sigma^{\le n}$ which uses $O\left({n}\alpha^{-1} \cdot\log \alpha
 \cdot \log n\right)$ bits.  Moreover, for any $\delta \in (0, 1)$,
 there is an inefficient $(1-\delta)$-accurate protocol for
 bounded $\alpha$-DTW$(\Sigma^{\le n})$ using space $O({n}\alpha^{-1}
   \cdot \log n + \log\delta^{-1})$ for any
 $\delta\in(0,1)$.
   \label{propboundedKDTW}
\end{prop}
\begin{proof}
  We assume without loss of generality that $\alpha$ and $n / \alpha$
  are integers. Let $x \in \Sigma^{\le n}$ be a string given to Alice,
  and $y \in \Sigma^{\le n}$ be a string given to Bob. Alice can
  construct a string $x'$ by taking each run in $x$ which is of length
  less than $\alpha$ and reducing its length to one. Notice that
  $\dtw(x', y) \le \dtw(x, y)$ trivially and that $\dtw(x, y) <
  \alpha\dtw(x', y)$ because any correspondence between $x'$ and $y$
  can be turned into a correspondence between $x$ and $y$ by
  duplicating every letter in the original correspondence $\alpha - 1$
  times. Thus if Alice could communicate $x'$ to Bob, then Bob could
  solve $\alpha$-DTW.

  In an attempt to communicate $x'$ to Bob, Alice first constructs a
  list $L$ consisting of the pairs $(i, l_i)$ for which the $i$-th run
  in $x$ is of length $l_i \ge \alpha$. Alice then sends $L$ to Bob. Notice
  that $|L| \le {n}/{\alpha}$, and thus can be communicated using
  $O(\frac{n}{\alpha} \log n)$ bits. Moreover, if Alice defines $x''$ to be
  $x$ except with every run reduced to length one, then $x'$ can be
  recovered from $x''$ and $L$. Therefore, if Alice could further
  communicate $x''$ to Bob, then Bob could solve $\alpha$-DTW.

  In an attempt to communicate $x''$ to Bob, Alice invokes the one-way
  communication protocol of \cite{edsketch} for the $3n/\alpha$-document
  exchange problem. She sends Bob the resulting sketch $s$ of size
  $O(n/\alpha \cdot\log \alpha \log n)$ bits which is correct with probability at
  least $p$. Bob, in turn, defines $y''$ to be $y$ with each run
  reduced to length one and uses the sketch $s$ along with $y''$ in
  order to try to recover $x''$. If Bob is able to use $s$ to recover
  a value for $x''$, then he can correctly solve $\alpha$-DTW with high
  probability. If, on the other hand, Bob is unable to use $s$ to
  recover a value for $x''$, then Bob may conclude with high
  probability that $\ed(x'', y'') > 3n/\alpha$. Because $\ed(x'', y'') \le
  3\dtw_0(x'', y'')$ by Lemma \ref{lem2approx} and because
  $\dtw_0(x'', y'') \le \dtw_0(x, y)$, we have that $n/\alpha < \dtw_0(x,
  y)$. It follows that in this case Bob can correctly return ``Fail''.
  
  Rather than using the efficient one-way communication protocol of
  \cite{edsketch}, Alice could instead invoke the protocol of
  \cite{orlitsky1991interactive} in which she sends Bob a hash of
  $x''$ using $O({n}\alpha^{-1} \cdot \log n + \log\delta^{-1})$ and
  Bob is able to then inefficiently recover $x''$ correctly with
  probability at least $1 - \delta$. Thus we also obtain an
  inefficient $(1 - \delta)$-accurate protocol which uses
  $O({n}\alpha^{-1} \cdot \log n + \log\delta^{-1})$ bits.
\end{proof}

\subsection{Protocol Over Well-Separated Tree Metrics}\label{subsectree}

In this Subsection, we generalize Proposition \ref{propboundedKDTW} to obtain
an efficient communication protocol for $\alpha$-DTW over
well-separated tree metrics. Before doing so, we provide a brief
background discussion for these metric spaces.
\begin{defn}
 Let $T$ be a tree whose nodes are the letters from the alphabet
 $\Sigma$ and whose edges have positive weights. Moreover, suppose
 that any root-to-leaf path of edges has non-increasing weights. Then
 define the distance between two nodes in $T$ to be the weight of the
 heaviest edge in the path between the two nodes. We call such a
 metric a \emph{well-separated tree metric}.

 If additionally the edges along every root-to-leaf path always
 decrease in weight by at least a factor of two between consecutive
 edges, then the metric is said to be a \emph{2-hierarchically
   well-separated tree metric}.
\end{defn}

Our definition differs slightly from the standard definition, which
defines the metric as simply being the graph distance metric induced
by the tree $T$ on its nodes. Importantly, these two definitions are
essentially equivalent (up to constant factor change in distances) for
the case of 2-hierarchically well-separated tree metrics. It was shown
in \cite{trees} that any finite metric $M$ can be embedded into a
2-hierarchically well-separated tree metric with expected distortion
$O(\log |M|)$. Thus well-separated tree metrics are in some sense
universal.

Next we define the notion of the $r$-simplification of a string. We
will then use $r$-simplifications to reduce $\alpha$-DTW to bounded
$\alpha$-DTW over well-separated tree metrics.
\begin{defn}[$r$-Simplification]
  Let $T$ be a well-separated tree metric whose nodes form the
  alphabet $\Sigma$. For a string $x \in \Sigma^{\le n}$ and for $r \ge 1$,
  we construct a string $s_r(x)$ by replacing each letter $l \in x$
  with the highest ancestor of $l$ in $T$ to be reachable from $l$ via
  only edges of weight at most $r / 4$. The string $s_r(x)$ is known
  as $x$'s \emph{$r$-simplification}.
\end{defn}

The next lemma states three important properties of $r$-simplifications.
\begin{lem}[Simplification Preserves DTW Gap]
  Let $T$ be a well-separated tree metric with distance function $d$
  and whose nodes form the alphabet $\Sigma$. Consider strings $x$ and
  $y$ in $\Sigma^{\le n}$.

  Then the following three properties of $s_r(x)$ and $s_r(y)$ hold:
  \begin{itemize}
  \item For every pair of distinct letters $l_1$ and $l_2$ in $s_r(x)$
    and $s_r(y)$, the distance $d(l_1, l_2)$ is greater than $r / 4$.
  \item If $\dtw(x, y) \le nr / \alpha$ then $\dtw(s_r(x), s_r(y)) \le nr / \alpha$ and $\dtw_0(s_r(x), s_r(y)) \le 4n / \alpha$.
  
  \item If $\dtw(x, y) > nr$, then $\dtw(s_r(x), s_r(y)) > nr/2$. 
  
  %\lin{Originally it was ``If $\dtw(x, y) > nr$, then $\dtw(s_r(x), s_r(y)) > nr / \alpha$.'', I think this is not sufficient since the bounded $\alpha$-DTW problem returns an $\alpha$-approximation.}
  %\item If $\dtw(x, y) > nr$, then $\dtw(s_r(x), s_r(y)) > nr / \alpha$. 
  \end{itemize}
  \label{lemthreeprops}
\end{lem}
\begin{proof}
  The first part of the claim follows immediately from the definition
  of $s_r(x)$ and $s_r(y)$.
  
  Consider a correspondence $C$ between $x$ and $y$ and define $C'$ to
  be the analogous correspondence between $s_r(x)$ and
  $s_r(y)$. Without loss of generality the correspondences are each of
  length at most $2n$.

  Consider any two letters $l_1 \in x$ and $l_2 \in y$ which form an
  edge in the correspondence $C$, and define $l_1'$ and $l_2'$ to be
  the corresponding letters in $s_r(x)$ and $s_r(y)$. If $d(l_1, l_2)
  \le r / 4$, then we will have $l_1' = l_2'$, meaning that $d(l_1',
  l_2') = 0$. If, on the other hand, $d(l_1, l_2) > r/4$, then the
  heaviest edge $e$ in the path from $l_1$ to $l_2$ in the tree $T$
  will also be the heaviest edge in the path from $l_1'$ to $l_2'$,
  meaning that $d(l_1', l_2') = d(l_1, l_2)$. Combining these two
  cases, it follows in general both that
  \begin{equation}
    d(l_1', l_2') \le d(l_1, l_2),
    \label{eqtree1}
  \end{equation}
  and that
  \begin{equation}
    d(l_1, l_2) \le d(l_1', l_2') + r/4.
    \label{eqtree2}
  \end{equation}

  By \eqref{eqtree1}, the cost of $C'$ is no larger than the cost of
  $C$. Therefore, $\dtw(s_r(x), s_r(y)) \le \dtw(x, y)$. Hence, if
  $\dtw(x, y) \le nr / \alpha$, then $\dtw(s_r(x), s_r(y)) \le nr /
  \alpha$.  Moreover, since each non-zero edge of the optimal DTW
  correspondence (which might not be an optimal $\dtw_0$
  correspondence) has cost at least $r/4$, it follows that
  $\dtw_0(s_r(x), s_r(y)) \le \dtw(s_r(x), s_r(y)) / (r/4) \le 4n /
  \alpha$, establishing the second part of the lemma.

  By \eqref{eqtree2}, the cost of $C$ exceeds the cost of $C'$ by at
  most $2n \cdot r/4 = nr/2$. Therefore, $\dtw(x, y) \le \dtw(s_r(x),
  s_r(y)) + nr / 2$. If $\dtw(x, y) > nr$, %then because $\alpha \ge 2$, 
  we
  get that $\dtw(s_r(x), s_r(y)) > nr/2$, establishing the third part
  of the lemma.
\end{proof}

Before we introduce the protocol for $\alpha$-DTW over a well-separated tree metric, we first define the following problem.
\begin{definition}[$(r, \alpha)$-gap DTW]
  In the \emph{$(r, \alpha)$-gap DTW$(\Sigma^{\le n})$} problem, Alice
  is given a string $x \in \Sigma^{\le n}$ and Bob is given a string
  $y \in \Sigma^{\le n}$. The valid solutions to the problem are $0$ if
    $\dtw(x,y)\le nr$ and $1$ if $\dtw(x,y) >  nr/\alpha$.
\end{definition}

Note that $(r, \alpha)$-gap DTW can be seen as simply being a
shorthand for the Distance Threshold Estimation Problem $\dtep_{rn /
  \alpha}^\alpha$. The properties of $r$-simplifications proven in
Lemma \ref{lemthreeprops} allow for the construction of a one-way
communication protocol for $(r, \alpha)$-gap DTW.

\begin{lem}[Protocol for $(r, \alpha)$-gap DTW Over Well-Separated Tree Metrics] 
  Suppose $\Sigma$ is a well-separated tree metric with both size and
  aspect ratio polynomial in $n$. Then for $p \in 1 -
  {\poly(n^{-1})}$, there is an efficient $p$-accurate one-way
  communication protocol for $(r, \alpha)$-gap DTW$(\Sigma^{\le n})$
  with complexity $O(n\alpha^{-1} \log
  \alpha \log n)$ and an inefficient $p$-accurate protocol with
  complexity $O(n\alpha^{-1} \log n)$.
  \label{lemgap}
\end{lem}
\begin{proof}
  We assume without loss of generality that $\alpha$ is at least a
  sufficiently large constant. In order to solve the $(r, \alpha)$-gap DTW problem for $x$ and $y$, we instead run a $p$-accurate protocol
  $\Pi$ for the bounded $\alpha/4$-$\dtw$ problem on $(s_r(x),
  s_r(y))$ for some $p = 1 - {\poly(n^{-1})}$.  We return ``1'' if the
  $\Pi$ returns ``Fail''.  Otherwise, $\Pi$ returns a number $Z$. We
  return ``1'' if $Z > nr/\alpha$ and ``0'' otherwise.

  To complete the proof, we condition on the event that this protocol
  is correct for the bounded $\alpha/4$-$\dtw$ problem, which happens
  with probability at least $p$, and then we show the correctness of
  the final output for the $(r, \alpha)$-gap DTW problem. 

  If $\Pi$ returns ``Fail'', then it must not be the case that
  $\dtw_0(s_r(x), s_r(y)) > 4n/\alpha$. Hence by Lemma \ref{lemthreeprops}, $\dtw(x, y)>
  nr/\alpha$, meaning that our protocol is correct in returning ``1''.

  If $\Pi$ does not return ``Fail'', then by definition of the
  $\alpha$-DTW problem, we have $\dtw(s_r(x), s_r(y))\le Z < \alpha \cdot
  \dtw(s_r(x), s_r(y))/4$.  By Lemma \ref{lemthreeprops}, if $Z >
  nr/\alpha$ then $\dtw(s_r(x), s_r(y)) > nr/\alpha$ and thus $\dtw(x,y)
  > nr/\alpha$; and if $Z \le nr/\alpha$, then $\dtw(s_r(x), s_r(y))
  \le nr/2$ and hence $\dtw(x, y) \le nr$.  In either case, the final
  output is correct for the $(r, \alpha)$-gap DTW problem. Moreover,
  by Proposition \ref{propboundedKDTW}, the efficient version of the
  protocol uses $O(n\alpha^{-1} \log \alpha \log n)$ bits and the
  inefficient version of the protocol uses $O(n \alpha^{-1} \log n)$
  bits.
\end{proof}

So far we have provided communication protocols for bounded
$\alpha$-DTW and $(r, \alpha)$-gap DTW. The next lemma shows how to
solve $\alpha$-DTW using a small number of instances of $(r,
\alpha)$-gap DTW and $\alpha$-bounded DTW.

\begin{lemma}[From Gap DTW to Approx.]
	\label{lemma:gap_to_approx}
	Let $\delta\in(0,1)$ be an error parameter and let $\kappa$
        denote the aspect ratio of a metric space $\Sigma$.  Suppose
        for every $r\ge 1$ and $1 \le \alpha\le n$ there is a
        $(1-\delta)$-accurate one-way protocol for $(r,
        \alpha)$-gap DTW$(\Sigma^{\le n})$ with space
        complexity $s_1(r, \alpha, \delta)$, and a
        $(1-\delta)$-accurate one-way protocol for bounded
        $\alpha$-DTW($\Sigma^{\le n}$) with space complexity
        $s_2(\alpha, \delta)$, then there is a
        $(1-\delta)$-accurate one-way protocol for
        $\alpha$-DTW$(\Sigma^{\le n})$ with space
        complexity \[ O\Bigg[\sum_{i=0}^{\lceil\log
            (2n\kappa \alpha)\rceil}s_1\bigg(2^i, \frac{\alpha}{2},
          \frac{\delta}{2\log (2n\kappa)}\bigg) + s_2\bigg[\alpha,
            \frac{\delta}{2}\bigg]\Bigg].
	\]

         Moreover, if the protocols for $(r, \alpha)$-gap DTW and
         bounded $\alpha$-DTW are efficient, then so is the protocol
         for $\alpha$-DTW.
\end{lemma}
\begin{proof}
	We begin with describing the full protocol. Without loss of generality the smallest distance in $\Sigma$ is $1$
	and the largest distance is $\kappa$. Therefore, $1\le \dtw(x,y)\le 2n\kappa$. In order to
	solve $\alpha$-DTW, Alice runs a $(1-\delta/2)$-accurate protocol $\wt{\Pi} =(\wt{\sk}, \wt{F})$ for the bounded $\alpha$-DTW  problem and  $(1-\frac{\delta}{2\log (2n\kappa \alpha)})$-accurate protocols $\Pi_i =(\sk_i, F_i)$ for the $(2^i, \alpha / 2)$-gap DTW problem on $x$ for
	each $i \in \{0, 1, \ldots,\lceil \log(2n\kappa \alpha) \rceil\}$. 
	Alice then sends $\wt{\sk}(x), {\sk}_1(x), {\sk}_2(x), \ldots, {\sk}_{\lceil \log(2n\kappa \alpha) \rceil}(x)$ to Bob. 
	If $\wt{F}(\wt{\sk}(x), y)$ is not ``Fail'', then Bob returns  $\wt{F}(\wt{\sk}(x), y)$ as the answer for $\alpha$-DTW.
	Otherwise, he finds the smallest $i$ such
	that the $(2^i, \alpha/2)$-gap DTW problem for $x$ and $y$ returns $0$, i.e., $F_i(\sk_i(x), y) = 0$.
	He then returns $2^in/\alpha$.
	
	Next we show the correctness of the protocol defined above.  We first condition each of the sub-protocols returning a correct answer, which by the union bound occurs with probability at least $1-\delta$.
	
	If $\wt{\Pi}$ does not return ``Fail'', then by definition of
        the bounded $\alpha$-DTW problem
        (Definition~\ref{def:bounded-alpha-dtw}), Bob obtains an
        $\alpha$-approximation for $\dtw(x, y)$.  On the other hand,
        if $\wt{\Pi}$ does return ``Fail'', it must be the case that
        $\dtw(x, y)> n/\alpha$.  In this case, Bob then finds the
        smallest $i$ such that the $(2^i, \alpha/2)$-gap DTW problem
        for $x$ and $y$ returns $0$.  Such an $i$ must exist since
        trivially $\dtw(x, y) \le 2n\kappa \le 2^{\log(2n\kappa\alpha)}n / \alpha$,
        which means that the case of $i = \lceil \log(2n\kappa\alpha)
        \rceil$ will return 0. After selecting the smallest such $i$,
        Bob then returns $2^i n / \alpha$.  If $i = 0$, since we also
        know that $\dtw(x, y) > n / \alpha$, Bob's returned answer of
        $n / \alpha$ will be an $\alpha$-approximation for $\dtw(x,
        y)$, as desired.  If $i > 0$, then the $(2^{i - 1}, \alpha /
        2)$-gap DTW problem tells us that $\dtw(x, y) > \frac{2^{i -
            1} n}{\alpha / 2} = 2^i n / \alpha$, and the $(2^i, \alpha
        / 2)$-gap DTW problem tells us that $\dtw(x, y) \le 2^i n$.
        Therefore Bob's returned answer of $2^i n / \alpha$ is again
        an $\alpha$-approximation for $\dtw(x, y)$.
	
	The space-complexity of the above protocol follows from a
        direct calculation, completing the proof.
\end{proof}

We can now solve $\alpha$-DTW over a well-separated tree metric:
\begin{lem}[Protocol for Well-separated Trees]
  Suppose $\Sigma$ is a well-separated tree metric with both size and
  aspect ratio polynomial in $n$. Then for $p \in 1 -
  {\poly(n^{-1})}$, there is an efficient $p$-accurate one-way
  communication protocol for $\alpha$-DTW over $\Sigma$ with
  complexity $O\left({n}{\alpha^{-1}} \cdot \log \alpha \cdot \log^2
  n\right)$ and an inefficient $p$-accurate protocol with space
  $O\left[{n}{\alpha^{-1}} \cdot \log^2  n\right]$.
 \label{thmDTWsketchspecializedtree}
\end{lem}
\begin{proof}
  Proposition \ref{propboundedKDTW} gives protocols for bounded $\alpha$-DTW which
  succeed with high probability and Lemma \ref{lemgap} gives protocols
  for $\alpha$-gap DTW which succeed with high probability. Plugging these
  into Lemma \ref{lemma:gap_to_approx} yields protocols for
  $\alpha$-DTW which succeeds with high probability. The bit
  complexities of both the efficient and inefficient variants of the
  resulting protocols follow directly from Lemma
  \ref{lemma:gap_to_approx}.
\end{proof}

\subsection{Upper Bound for Finite Metrics}\label{subsectreetogeneral}
With the help of Lemma \ref{thmDTWsketchspecializedtree}, we are now
prepared to prove a more general theorem. For arbitrary finite metrics
satisfying certain natural constraints, the $2/3$-accurate one-way
communication complexity of $\alpha$-DTW is within a polylogarithmic
factor of ${n}/{\alpha}$. (The analagous lower bound will appear in
Section \ref{subsecmainlow}.)
\begin{thm}[Theorem \ref{thm:overviewmain} restated]

  Let $\Sigma$ be a metric space
  of size and aspect ratio polynomial in $n$.  Then there is an efficient
  $2/3$-accurate one-way communication protocol for $\alpha$-DTW over
  $\Sigma$ with space complexity $O\left({n}{\alpha^{-1}}\cdot \log
  \alpha \cdot\log^3 n\right)$ and an inefficient $2/3$-accurate
  one-way protocol with complexity
  $O\big({n}{\alpha^{-1}}\cdot \log^3 n)$.
  \label{thmDTWsketchgeneral}
\end{thm}
\begin{proof}
  It is shown in \cite{trees} that $\Sigma$ can be embedded into a
  2-hierarchically well-separated tree metric $\Sigma'$ using a
  randomized map $\phi$ so that for any $a, b \in \Sigma$, $d(a, b)
  \le d(\phi(a), \phi(b))$ and $\E(d(\phi(a), \phi(b))) \le O(\log n)
  \cdot d(a, b)$. (Here $d$ is taken to be the appropriate distance
  function over either $\Sigma$ or $\Sigma'$.)

  Consider two strings $x, y \in \Sigma^{\le n}$, and let $x'=\phi(x)$
  and $y'=\phi(y)$ denote $x$ and $y$ with their letters mapped into
  $\Sigma'$ by $\phi$. For any correspondence $C$ between $x$ and $y$,
  and the analogous correspondence $C'$ between $x'$ and $y'$, we
  have that
  \begin{equation}
    \text{cost}(C) \le \text{cost}(C'),
    \label{eqcost1}
  \end{equation}
  and
  \begin{equation}
    \E[\text{cost}(C')] \le O(\log n) \cdot \text{cost}({C}).
    \label{eqcost2}
  \end{equation}

  If we select $C'$ to be an optimal correspondence between $x'$
  and $y'$, then it follows from \eqref{eqcost1} that
  $$\dtw(x, y) \le \text{cost}(C)\le \text{cost}(C')= \dtw(x', y').$$

  On the other hand, if we select $C$ to be an optimal correspondence
  between $x$ and $y$, then it follows from \eqref{eqcost2} that
  $$\E[\dtw(x', y')]\le \E[\text{cost}(C')] \le \Theta(\log
  n) \cdot \text{cost}(C) = \Theta(\log n) \cdot \dtw(x, y).$$

  By Markov's inequality, with probability at least $9/10$,
  $$\dtw(x, y) \le \dtw(x', y') \le O(\log n) \cdot \dtw(x, y).$$

  Applying Lemma \ref{thmDTWsketchspecializedtree} with $p = 9/10$
  to $\Sigma'$, we get an efficient $(9/10)^2$-accurate one-way communication
  protocol for $O(\alpha \log n)$-DTW using
  $$O\left(\frac{n}{\alpha} \log \alpha \log^2 n\right)$$
  bits. Defining $\alpha'$ as $\Theta(\alpha \log n)$, we get an
  efficient $(9/10)^2$-accurate one-way communication protocol for
  $\alpha'$-DTW using
  $$O\left(\frac{n}{\alpha'} \log \alpha' \log^3 n\right)$$ bits. This
  completes the construction of an efficient protocol for the case
  where $\alpha$ is at least $\Omega(\log n)$, and the case of $\alpha
  \in O(\log n)$ follows by simply having Alice send all of $x$ as her
  sketch.
  
  The construction of the inefficient sketch follows similarly,
  completing the proof.
\end{proof}

%% Note that Theorem \ref{thmDTWsketchgeneral} can be applied to any
%% metric space $\Sigma$ satisfying each of the following properties:
%% \begin{enumerate}
%% \item $3 \le |\Sigma|$.
%% \item $|\Sigma| \le \poly(n)$.
%% \item The ratio between the smallest distance between letters in
%%   $\Sigma$ and the largest distance between letters in $\Sigma$ is at
%%   most polynomial in $n$.
%% \end{enumerate}

%% The remainder of the section is devoted to showing that each of these
%% properties are necessary for the Theorem to hold. [TODO: To show the
%%   first bullet point, I have an not-yet-written-up sketch using
%%   polylogarithmic number of bits for the case of $|\Sigma| = 2$; I'm
%%   not sure what the best way is to show that the second bullet point
%%   is required; To show the third bullet point, I have a
%%   not-yet-written-up lower bound of $\Omega(n)$ bits for performing
%%   communicating $K/2$-DTW over the alphabet $1, K, K^2, \ldots,
%%   K^{2n}$.]

\subsection{Further Optimizing Upper Bounds}\label{subsecgeneralgeneral}

Recall that Theorem \ref{thmDTWsketchgeneral} provides a one-way
communication protocol for $\alpha$-DTW over $\Sigma^{\le n}$ which is
within a polylogarithmic factor of optimal. We will now show that for
several important cases of $\Sigma$, the bound from Theorem
\ref{thmDTWsketchgeneral} can be improved by roughly a logarithmic
factor. In order to do this, we first introduce the notion of an
efficiently $\sigma$-separable metric space.

\begin{defn}
Suppose we have a metric space $(\Sigma, d)$. A partition of $\Sigma$
is \emph{$\delta$-bounded} if the diameter of each part of the
partition is at most $\delta$. A probability distribution over
$\delta$-bounded partitions of $\Sigma$ is said to be
\emph{$\sigma$-separating} if for all $x,y \in \Sigma$, the
probability that $x$ and $y$ are in different parts of the partition
is at most $\sigma \cdot d(x, y) / \delta$. Finally, $\Sigma$ is said
to be \emph{efficiently $\sigma$-separable} if, for every $\delta >
0$, there exists a $\sigma$-separating probability distribution over
$\delta$-bounded partitions of $\Sigma$, and if a partition can be
selected from the $\sigma$-separating probability distribution in time
$\poly(|\Sigma|)$.
\end{defn}

The main result in this subsection is the following:
\begin{thm}[Theorem \ref{thm:overviewextended} restated]
   Let $(\Sigma, d)$ be a metric space of size and aspect ratio
   $\poly(n)$. Suppose that $(\Sigma, d)$ is efficiently
   $\sigma$-separable for some $1 \le \sigma \le O(\log n)$. Then
   there is an efficient $2/3$-accurate one-way communication protocol
   for $\alpha$-DTW$(\Sigma^{\le n})$ with space complexity
   $$O\left({\sigma n}{\alpha^{-1}}\cdot \log \alpha \cdot\log^2
   n \cdot \log \log \log n\right)$$ and an inefficient $2/3$-accurate one-way protocol with
   space complexity $$O\big({\sigma
     n}{\alpha^{-1}}\cdot\log^2 n \cdot \log \log \log n\big).$$
  \label{thmmainthmgeneralized}
\end{thm}
\begin{proof}
  See Section \ref{secappendix2} at the end of the paper.
\end{proof}

The advantage of Theorem \ref{thmmainthmgeneralized} is that many
metric spaces are known to be $\sigma$-separable for small
$\sigma$. When this is the case, Theorem \ref{thmmainthmgeneralized}
can be used in place of Theorem \ref{thmDTWsketchgeneral} to replace a
factor of $\log n$ in the space complexity with a factor of $\sigma \cdot \log \log \log n$.

Note that, in general, any metric space $M$ of polynomial size is efficiently $O(\log
|M|)$-separable \cite{bartal1996probabilistic}, meaning that Theorem
\ref{thmmainthmgeneralized} implies a general bound within a factor of
$\log \log \log n$ of Theorem \ref{thmDTWsketchgeneral}.

As an important special case, for $(\Sigma, d) \subseteq
(\mathbb{R}^d, \ell_p)$ (i.e. points in $d$-dimensional space equipped
with $\ell_p$ norm), it is known that $(\Sigma, d)$ is
$O(d^{1/p})$-separable for $p \in [1,2]$ and is $O(\sqrt{d \cdot
  \min(p, \log d)})$-separable for $p \ge 2$
\cite{charikar1998approximating, naor2017probabilistic}. Thus we get
the following corollary of Theorem \ref{thmmainthmgeneralized}.

\begin{cor}\label{cor:414}
  Let $(\Sigma, d)$ be a metric space of size and aspect ratio
  polynomial in $n$. Moreover, suppose that $(\Sigma, d) \subseteq
  (\mathbb{R}^d, \ell_p)$.

  If $p \in [1, 2]$, then there is an efficient $2/3$-accurate one-way
  communication protocol for $\alpha$-DTW$(\Sigma^{\le n})$ with space
  complexity $O\left({d^{1/p} n}{\alpha^{-1}}\cdot \log \alpha
  \cdot\log^2 n \cdot \log \log \log n\right)$ and an inefficient $2/3$-accurate one-way
  protocol with space complexity $O\big({d^{1/p}
    n}{\alpha^{-1}}\cdot\log^2 n \cdot \log \log \log n\big)$.

  If $p > 2$, then there is an efficient $2/3$-accurate one-way
  communication protocol for $\alpha$-DTW$(\Sigma^{\le n})$ with space
  complexity $O\left({\sqrt{d \cdot \min(p, \log d)}
    n}{\alpha^{-1}}\cdot \log \alpha \cdot\log^2 n \cdot \log \log \log n\right)$ and an
  inefficient $2/3$-accurate one-way protocol with space complexity
  $O\big({\sqrt{d \cdot \min(p, \log d)}
    n}{\alpha^{-1}}\cdot \log^2 n \cdot \log \log \log n\big)$.
\end{cor}

Additionally, if $(\Sigma, d)$ is a metric space of polynomial size
with doubling constant $\lambda$ (recall the doubling constant of a metric
$(\Sigma, d)$ is $\lambda$ if for all $x \in \Sigma$ and $r > 0$, the
ball $B(x, 2r)$ can be covered by $\lambda$ balls of radius $r$), 
then it will be efficiently $O(\log
\lambda)$-separable \cite{neimanstochastic}, yielding the following
corollary of Theorem \ref{thmmainthmgeneralized}.

\begin{cor}\label{cor:doubling}
  Let $(\Sigma, d)$ be a metric space of size and aspect ratio
  polynomial in $n$, and with doubling constant $\lambda$.

  There is an efficient $2/3$-accurate one-way communication protocol
  for $\alpha$-DTW$(\Sigma^{\le n})$ with space complexity
  $O\left({\log \lambda \cdot n}{\alpha^{-1}}\cdot \log \alpha
  \cdot\log^2 n \cdot \log \log \log n\right)$ and an inefficient
  $2/3$-accurate one-way protocol with space complexity $O\big({\log
    \lambda \cdot n}{\alpha^{-1}}\cdot \log^2 n \cdot \log \log \log
  n\big)$.
\end{cor}

\section{Lower Bounds}\label{sec:lb}

In this section, we present lower bounds for the one-way communication
complexity of $\alpha$-DTW. In Subsection \ref{subsecmainlow}, we show
that as long as $|\Sigma| \ge 3$, then regardless of the distance
function $d$, the one-way communication complexity of $\alpha$-DTW is
$\Omega(n / \alpha)$. Somewhat surprisingly, this result holds even
when the distance function $d$ does not satisfy the triangle
inequality (i.e., $(\Sigma, d)$ need not be a metric space).

When $\Sigma$ is a metric space of polynomial size and aspect ratio,
the lower bound of $\Omega(n / \alpha)$ is within a polylogarithmic
factor of tight. In Subsection \ref{secdtw0low}, we show that for the
special case of $\dtw$ over generalized Hamming space, the lower bound
can be improved to $\Theta(n/\alpha \log n)$ bits for $\mathrm{CC}_{1
  - 1/n}(\alpha\text{-}\dtw(\Sigma^{\le n}))$, which is within a constant
factor of tight.

Finally, in Subsection \ref{subseclinearsketch}, we turn our attention
to the more restrictive model of linear sketches. We show that no
linear sketch can solve $\alpha$-DTW over $\{0, 1, 2\}^{n}$ with fewer
than $\Omega(n)$ bits.

\subsection{Lower Bound Over Arbitrary Alphabets}\label{subsecmainlow}

For ease of exposition, we begin by considering the case of $\Sigma =
\{0, 1, 2\}$.
\begin{theorem}
  Consider $1 \le \alpha \le n$. For
  $\Sigma = \{0, 1, 2\}$,
  $$\text{CC}_{0.1}\big[\dtep_{1}^{\alpha}(\Sigma^{\le
      n})\big] = \Omega(n/\alpha).$$
        \label{thmlowerint}
\end{theorem}
\begin{proof}
Without loss of generality, we assume $\alpha$ and $n/\alpha$ are
positive integers. Let $(\sk, F)$ be a $0.9$-accurate one-way
communication protocol for $\alpha/2$-DTW over $\{0, 1, 2\}^{\le 3n}$.\footnote{Note that we could just as well use a protocol for $c\alpha$-DTW for any $c < 1$. We consider $\alpha / 2$-DTW for convenience.}

We prove the lower bound by reduction from INDEX$_{n/\alpha}$. We
begin with the description of the protocol. Suppose Alice has a length
$k=n/\alpha$ binary string $x=(x_1, x_2, \ldots, x_k)$, and Bob has an
index $i\in[k]$. The INDEX$_{n / \alpha}$-problem requires Bob to
recover $x$'s $i$-th letter $x_i$. It is well known that the
$.9$-accurate one-way communication complexity of INDEX$_{n/\alpha}$
is $\Theta(n/\alpha)$ \cite{knr99}. We now present a communication
protocol for INDEX$_{n / \alpha}$ in which Alice's message is
constructed using the sketch $\sk$ for $\dtw$ over $\{0, 1,
2\}^{\le 3n}$, thereby establishing that the bit-complexity of $(\sk,
F)$ is at least $\Omega(n / \alpha)$.

Define
\begin{align*}
Z(1) &= \big(\underbrace{0, 0, \ldots, 0}_{\alpha~\text{times}}, \underbrace{1, 1, \ldots, 1}_{\alpha~\text{times}}, \underbrace{2, 2, \ldots, 2}_{\alpha~\text{times}}\big)\quad\text{and}\\
Z(0) &= \big(\underbrace{0, 0, \ldots, 0}_{\alpha~\text{times}}, 1, \underbrace{2, 2, \ldots, 2}_{\alpha~\text{times}}\big).
\end{align*}
Alice constructs the string $A(x)$ defined by,
\[
A(x) = Z(x_1)\circ Z(x_2) \ldots \circ Z(x_k),
\]
and sends Bob the sketch $\sk(A(x))$. Bob then constructs a string
$B(i)$ defined by,
\[
B(i) = \underbrace{(0,1,2) \circ  \cdots \circ (0,1,2)}_{(i-1)~\text{times}} \circ (0,2)\circ\underbrace{(0,1,2) \circ \cdots \circ (0,1,2)}_{(k-i)~\text{times}}.
\]
We will prove that if $x_i = 0$, then $\dtw(A(x), B(i)) \le 1$ and
that if $x_i = 1$, then $\dtw(A(x), B(i)) \ge \alpha$. Thus Bob will
be able to use Alice's sketch to correctly determine $x_i$ with
probability at least $0.9$, completing the proof.

\paragraph{Case 1: }First consider the case where $x_i = 0$. We wish to show that $\dtw(A(x), B(i)) \le 1$. Consider a correspondence
$(\wb{a}, \wb{b})$ between $A(x)$ and $B(i)$ that maps each of Alice's
$Z(x_{j})$'s to the $j$-th $(0,1,2)$ in Bob's string for all $j\neq
i$, and that maps Alice's $Z(x_{i}, i)$ to Bob's $(0,2)$. In
particular, this can be accomplished by defining $\wb{a} = a$ and
\[
\wb{b} = \underbrace{Z(1)\circ Z(1) \circ Z(1)}_{(i-1)~\text{times}}\circ  W(x_i)\circ\underbrace{Z(1)\circ Z(1) \circ Z(1)}_{(k-i)~\text{times}},
\]
where $W(x_i)$ is the expansion of $(0,2)$ to $0^{\alpha+1} \circ
2^{\alpha}$. Because $\wb{a}$ and $\wb{b}$ have only a single Hamming
difference (the $(\alpha + 1)$-th $0$ in $W(x_i)$ is matched with a
zero), we get that
\[
\dtw(A(x), B(i)) \le \|\wb{a} - \wb{b}\|_1 =  |0 -1| = 1,
\]
as desired.

\paragraph{Case 2: }Next consider the case of $x_i = 1$. We wish to show that any correspondence $(\wb{a}, \wb{b})$ of $(A(x), B(i))$ will cost at least $\alpha$. Suppose for contradiction that there exists a correspondence
$(\wb{a}, \wb{b})$ such that $\|\wb{a}- \wb{b}\|_1 < \alpha$. In
order so that $\|\wb{a}- \wb{b}\|_1 < \alpha$, it must be that each
of the $0$s in $B(i)$ are matched by the correspondence to $0$s from
at most one of $A(x)$'s runs. Otherwise, the $0$ in $B(i)$ would have
to also be matched with at least $\alpha$ $2$s, a
contradiction. Moreover, each of the $\alpha$-letter runs of $0$s in
Alice's string $A(x)$ must be matched with at least one $0$ from Bob's
string $B(i)$, since otherwise we would again necessarily have
$\|\wb{a}- \wb{b}\|_1 \ge \alpha$. Because each $0$ in Bob's string
matches with at most one run of $0$s in Alice's string, and each run
of $0$s in Alice's string is matched with at least one $0$ in Bob's
string, it follows that for each $j$ the $j$-th run of $0$s in Alice's
string matches with the $j$-th $0$ in Bob's string. However, this
prevents the $1$s from Alice's $Z(x_i)$ from being matched to any
$1$s from Bob's $B(i)$, forcing $\dtw(A(x), B(i)) \ge \alpha$, a contradiction.

\end{proof}

As we shall see in a moment, the proof of Theorem \ref{thmlowerint}
can be used without modification to prove a far more general
theorem. We also remark that for the special case of $\Sigma \subseteq
\mathbb{Z}$, a slightly stronger version of Theorem \ref{thmlowerint}
can be proven. (See Section \ref{secappendixlower} at the end of the paper.)

Notice that the preceding lower bound does not rely on the properties of
the numbers $0$, $1$, and $2$ beyond the fact that $d(0, 1)$ is the
smallest pairwise distance between the points. Consequently, the bound
generalizes to any three-letter space $\Sigma = \{a, b, c\}$ armed
with a two-point function $d:\Sigma\times \Sigma \rightarrow \RR_{+}$,
not necessarily satisfying the triangle inequality.  Without loss of
generality, we may assume
\[
r:=d(a, b)\le d(b, c) \le d(a, c)
\]
for some constant $r > 0$. The following theorem is then implied by the
same proof as Theorem \ref{thmlowerint}, except with $0, 1, 2$
replaced with $a, b, c$.
\begin{theorem}[Theorem \ref{thm:overviewgenerallower} restated]
	Let $\Sigma = \{a, b, c\}$ be three letters with a two-point
        function $d:\Sigma\times \Sigma \rightarrow \RR_{+}$. Consider
        $1 \le \alpha \le n$.
        Then CC$_{0.1}\big[\dtep_{r}^{\alpha}(\Sigma^{\le n})\big] =
        \Omega(n/\alpha)$.
        \label{thmmainlow}
\end{theorem}

%% As a corollary, we show that the lower bound applies to the $p$-DTW distance by applying the above theorem to $\dtw_p^p$.
%% \begin{corollary}
%% 	Let  $\Sigma = \{a, b, c\}$ be three letters of a metric space with distance function $d(\cdot, \cdot)$.
%% 	Suppose $0<d(a, b)\le d(b, c) \le d(a, c)$ and constant $p>0$.
%% 	Fix $\alpha \ge 1$, and $n\in \ZZ_{+}$ as parameters.
%% 	Then CC$_{0.1}\big[\alpha\text{-}\dtw_p(\Sigma^{\le n})\big] = \Omega(n/\alpha^p)$.
%% \end{corollary}

\subsection{Lower Bound For of DTW$_0$}\label{secdtw0low}
The following theorem establishes a tight bound for the one-way
communication complexity of $\dtw$ over generalized Hamming space.
\begin{thm}[Theorem \ref{thmhammingfirst} restated]
	Consider $1 \le \alpha \le n$, and consider the generalized
        Hamming distance over a point-set $\Sigma$ with $\Sigma$ of
        polynomial size $n^{1+\Omega(1)}$.  For $p \ge 1 -
        1/|\Sigma|^{-1}$, the $p$-accurate one-way communication
        complexity of $\alpha$-DTW$(\Sigma^{\le n})$ is
        $\Theta[n\alpha^{-1}\cdot\log n]$.
	\label{thmgenhamcomplexity}
\end{thm}

%Lemma \ref{lemboundedKDTW} upper bounded $C$ by $O\big({n}{\alpha^{-1}}\cdot
%\log \alpha\cdot \log n \big)$. The proof of the lemma built upon the
%$\alpha$-document exchange communication protocol of \cite{edsketch}. If
%instead one uses the $\alpha$-document exchange communication protocol of
%\cite{orlitsky1991interactive} in which Alice simply hashes her string
%to a $\Theta(n / \alpha \cdot\log n)$-bit pair-wise independent hash, then the
%bound becomes
%$$C \le O\left({n}{\alpha^{-1}} \cdot\log n\right),$$ as desired (though the
%communication protocol is no longer efficient).

Proposition \ref{propboundedKDTW} implies the desired upper bound (the
inefficient protocol). In order to prove Theorem
\ref{thmgenhamcomplexity}, it therefore suffices to prove the lower
bound. To do this, we first introduce a problem with high one-way
communication complexity.
\begin{lem}[\cite{jw13}, Theorem 3.1]
	Let the problem $(n, \mathcal{S})$-SET be defined as
        follows. Alice gets an $n$-element set $S \subseteq
        \mathcal{S}$ and Bob gets a character $a \in \mathcal{S}$. The
        goal is for Bob to determine whether $a \in S$.  Let $p \ge 1
        - \frac{1}{|\mathcal{S}|}$. Then the $p$-accurate one-way
        communication complexity of $(n, \mathcal{S})$-SET is $\Omega(n\log
        (|S|/n))$.
	\label{lemsetcomplexity}
\end{lem}

We are now prepared to prove Theorem \ref{thmgenhamcomplexity}.
\begin{proof}[Proof of Theorem \ref{thmgenhamcomplexity}]
	As described above, it suffices to show the lower bound. To
        this end, we reduce $(n/\alpha, \Sigma)$-SET to $\alpha$-DTW
        for strings of length $n$. Suppose Alice is given $S \subseteq
        \Sigma$ of size $n/\alpha$ and Bob is given the character $a
        \in \Sigma$. Then Alice can compute $x$ to be the
        concatenation of the elements of $S$ in an arbitrary order.
        Alice will use the resulting string $x \in \Sigma^{\le n}$ as
        an input for the $\alpha$-DTW problem.  Bob can then define
        $y$ to be the character $a$ repeated $n$ times. Notice that if
        $a \in S$ then $\dtw(x, y) = n/\alpha - 1$, whereas if $a
        \not\in S$ then $\dtw(x, y) = n$. By Lemma
        \ref{lemsetcomplexity}, this reduction establishes that for $p
        \ge 1 - \frac{1}{|\Sigma|}$ the $p$-accurate one-way
        communication complexity of $\alpha$-DTW is at least
	$$\Omega(n\alpha^{-1}\cdot\log(\alpha|\Sigma|/n)) =
        \Omega(n\alpha^{-1}\cdot\log n),$$ where the last equality
        holds since $|\Sigma| = n^{1+ \Omega(1)}$.
	%Standard
	%amplification arguments extend the bound to hold for any $p = 1-
	%\frac{1}{\poly(n)}$.
\end{proof}

\subsection{Lower Bound for Linear Sketching}\label{subseclinearsketch}
In this section, we establish a lower bound for linear sketching. A
similar lower bound is studied by Andoni, Goldberger, McGreger \&
Porat in \cite{AGMP13} for a variant of edit distance.
\begin{definition}
For $\Sigma \subseteq \mathbb{R}$, a $\delta$-error linear sketch for
$\alpha$-DTW over $\Sigma^{n}$ is a randomized function $\sk:
\Sigma^{\le n} \rightarrow \mathcal{V}^m$ for some vector space
$\mathcal{V}$ and dimension $m$ such that
\begin{enumerate}
\item For $x, y \in \sigma^n$, the pair $(\sk(x), \sk(y))$ can be used to solve $\alpha$-DTW (without additionally examining $x$ or $y$) with probability at least $1 - \delta$;
\item If $x, y, z \in \Sigma^{\le n}$ satisfy $x + y = z$, then  $\sk(x) + \sk(y) = \sk(z)$.
\end{enumerate}
The \emph{space complexity} of the sketch $\sk$ is given by $m \cdot
\log |\mathcal{V}|$.
\end{definition}
%Here all the arithmetic operations are in $\$.
%\subsubsection{A  Lower Bound for Linear Sketch: Ruling Out Compression Using Linear Sketch}
%Let $(L, F)$ be an $(\alpha, 0.1)$-linear sketch for DTW.
Next we prove an $\Omega(n)$ lower bound on the space complexity of
linear sketches.  For DTW over the alphabet $\{0, 1, 2\}^n$, this
shows that no linear sketch can achieve compression by more than a
constant factor.

\begin{theorem}[Theorem \ref{thm:overviewlinear} restated]
	Consider $1 \le \alpha \le n$.  Then any $0.1$-error
        linear sketch for $\alpha$-$\dtw$ on $\{0,1, 2\}^{4n}$ has
        space complexity $\Omega(n)$.
        \label{thmlinear}
\end{theorem}
\begin{proof}
	Let $\sk$ be an $0.1$-error linear sketch for $\alpha$-DTW
        on $\{0, 1\}^{4n}$. Recall that in the $\text{INDEX}_n$
        problem, Alice has a string $(x_1, x_2, \ldots, x_n) \in
        \{0,1\}^n$, and Bob has an index $i \in [n]$. Alice sends Bob
        a single round of communication and then Bob must recover
        $x_i$. The INDEX lower bound states that any protocol for this
        problem which succeeds with probability at least $0.9$ must
        use $\Omega(n)$ bits of communication \cite{knr99}. We now
        present a protocol for $\text{INDEX}_n$ in which Alice
        constructs her message to Bob using the linear sketch $\sk$ for
        $\dtw$ over $\{0, 1, 2\}^{4n}$, thereby establishing that the
        bit-complexity of $\sk$ is at least $\Omega(n)$.

	\paragraph{The Protocol} Alice constructs her message for Bob by computing the sketch $\sk(\bar{x})$, where, 
	\[
	\bar{x}=(1, x_1, x_1, 1, 1, x_2, x_2, 1, \ldots, 1, x_i, x_i,
        1,  \ldots, 1, x_n, x_n, 1).
	\]
	Bob then constructs a sketch for the vector $\bar{y}_i$, defined by,
	\[
	\bar{y}_i = (1, x_1, x_1, 1, 1, x_2, x_2, 1, \ldots, 1, x_i, \boldsymbol{x_i+1}, 1,  \ldots, 1, x_n, x_n, 1),
	\]
	using the identity
	\[
	\sk(\bar{y}_i)
	= \sk(\bar{x}) + \sk(0^{4(i-1)}, 0, 0, 1, 0 ,0^{4(n-i)}).
	\]
	If $x_i=0$, then $\dtw(\bar{x}, \bar{y}_i)= 0$ and if $x_i=1$
        then $\dtw(\bar{x}, \bar{y}_i)= 1$. Therefore any
        multiplicative approximation can distinguish the two cases,
        and Bob can use $\sk(\bar{x})$ and $\sk(\bar{y_i})$ to determine
        whether $x_i = 0$ or $x_i = 1$ with probability at least
        $0.9$. This completes the protocol.
\end{proof}

\section{Acknowledgments} D. Woodruff would like to thank NSF Big Data grant
1447639 for support for this project; part of this work was also done while the author was visiting the Simons Insitute for the Theory of Computing. W. Kuszmaul would like to thank
support from NSF grants 1314547 and 1533644, as well as a Hertz Fellowship and an
NSF GRFP Fellowship.

\bibliographystyle{alpha}
\bibliography{ref}

\appendix

\section{Improved Lower Bound Over Integers}\label{secappendixlower}
\begin{theorem}
	Consider $1 \le \alpha \le n$. Let $\Sigma=\{0,1,2,\ldots, m\}$ with $2\le m\le
        \alpha$, and let $\delta \le (m+1)^{-1}$.  Then
        CC$_{\delta}\big[\dtep_{1}^{\alpha}(\Sigma^{\le
            \Theta(n)})\big] = \Omega(n / \alpha \cdot \log m)$.
        \label{thmintlower}
\end{theorem}
%This lower bound is stronger than the one in the last section, since it shows that the alphabet needs not to be dependent on $n, \alpha$.
\begin{proof}
	Without loss of generality, we assume $\alpha$ and $n/\alpha$
        are positive integers.  We prove the lower bound by reduction
        from indexing over a large alphabet. Suppose Alice has a
        length $k=n/\alpha$ string $x=(x_1, x_2, \ldots, x_k)\in
        \Sigma\backslash\{0,m\}$, and that Bob has both an index
        $i\in[k]$ and a number $y_i\in \Sigma\backslash\{0,m\}$. Bob
        wishes to determine whether $x_i = y_i$. The $(1 - (m +
        1)^{-1})$-accurate one-way communication complexity of this
        problem is $\Omega(k \cdot \log |\Sigma|) = \Omega(n / \alpha
        \cdot \log m)$ \cite{jw13}. We now reduce this indexing
        problem to $\dtep_1^\alpha$.

        Alice converts the $j$-th character $x_j$ to a string $Z(x_j)$
        of length at most $n(m+3\alpha-3)/\alpha$, by defining
	\begin{align*}
	Z(x_j) &:= \big(0^{\alpha}, 1, 2, \ldots, x_j-1, x_j^{\alpha}, x_j +1,  \ldots, m-1, m^{\alpha}\big),
	\end{align*}
	where for a letter $l$ and length $t$, $l^{t}$ represents $t$ copies of $l$ concatenated together.
	Alice then constructs $A(x):=Z(x_1)\circ Z(x_2)\circ \ldots \circ Z(x_k)$. 
	Bob constructs \[
	B(i):=[0,m]^{i-1}\circ Y(y_i)\circ [0,m]^{k-i},
	\]
	where $[0, m]^{t}$ denotes the string $(0, 1, \ldots, m)$ repeated $t$ times, and $Y(y):= (0,1,\ldots, y-1, y+1, \ldots, m)$.
	
	Let us consider the case $x_i \neq y_i$.
	In this case, we claim $\dtw(A(x), B(i)) \le 1$. 
	To show this we can first find a correspondence that maps every $ Z(x_{i'})$ to $[0, m]$ with cost $0$ except for the $i$-th block. 
	Moreover, it is easy to see that $\dtw(Z(x_i), Y(y_i))\le 1$. 
	Therefore, there exists an correspondence that has cost at most $1$.

	Next, if $x_i = y_i$, the claim is that, for any
        correspondence $(\wb{a}, \wb{b})$ of $(A(x), B(i))$ we have
        $\|\wb{a}- \wb{b}\|_1 \ge \alpha$.  Suppose for contradiction
        that there exists a correspondence $(\wb{a}, \wb{b})$ such
        that $\|\wb{a}- \wb{b}\|_1 < \alpha$. In order so that
        $\|\wb{a}- \wb{b}\|_1 < \alpha$, it must be that each of the
        $0$s in $B(i)$ are matched by the correspondence to $0$s from
        at most one of $A(x)$'s runs. Otherwise, the $0$ in $B(i)$
        would have to also be matched with at least $\alpha$ $m$s, a
        contradiction.  Moreover, each of the $\alpha$-letter runs of
        $0$s in Alice's string $A(x)$ must be matched with at least
        one $0$ from Bob's string $B(i)$, since otherwise we would
        again necessarily have $\|\wb{a}- \wb{b}\|_1 \ge
        \alpha$. Because each $0$ in Bob's string matches with at most
        one run of $0$s in Alice's string, and each run of $0$s in
        Alice's string is matched with at least one $0$ in Bob's
        string, it follows that for each $j$ the $j$-th run of $0$s in
        Alice's string matches with the $j$-th $0$ in Bob's string
        $B(i)$. However, this prevents the run of $\alpha$ $x_i$s in
        Alice's $Z(x_i)$ from being matched with any $x_i$'s from
        Bob's string $B(i)$, forcing $\dtw(A(x), B(i)) \ge \alpha$, a
        contradiction.
\end{proof}

\section{Proof of Theorem \ref{thmmainthmgeneralized}}\label{secappendix2}

To prove Theorem \ref{thmmainthmgeneralized}, we first prove the
following lemma.

\begin{lemma}[Protocol of $(r,\alpha)$-gap DTW]
  \label{lemma:r_a_gap_general}
Let $(\Sigma, d)$ be a finite metric space of at most polynomial
size. Suppose that $(\Sigma, d)$ is efficiently $\sigma$-separable for
some $1 \le \sigma \le O(\log n)$. Then for any $\delta \in (0,1)$ and
$r \ge 1$, there exits an efficient $(1-\delta)$-accurate one-way
protocol for the $(r,\alpha)$-gap $\dtw(\Sigma^{\le n})$ problem with
space complexity $O(n\alpha^{-1} \sigma \log \alpha \cdot \log n \log
\delta^{-1})$ and an inefficient $(1-\delta)$-accurate one-way
protocol with space complexity $O(n\alpha^{-1} \sigma \log \alpha
\cdot \log n \cdot \log \delta^{-1})$.
\end{lemma}
\begin{proof}
Suppose Alice is given $x = (x_1, \ldots, x_{n_1}) \in \Sigma^{\le n}$
and Bob is given $y = (y_1, \ldots, y_{n_2}) \in \Sigma^{\le n}$.
Without loss of generality, we assume $\alpha \ge c \sigma$ for some
sufficiently large constant $c$, since otherwise Alice can send $x$
directly using $O(n \log n) \le O(n / \alpha \cdot \log^2 n)$ bits. In
what follows, we describe an efficient $0.51$-accurate protocol for
$(r, \alpha)$-gap DTW. The boost of probability to $1-\delta$ follows
from standard amplification arguments.

Using the fact that $(\Sigma, d)$ is efficiently $\sigma$-separable,
we first construct a partition $P$ of $\Sigma$ with diameter $l$ for
some $l=c_1r$ for $c_1$ a sufficiently large constant. The partition
$P$ satisfies the property that for $x, y \in \Sigma$, the probability
that $x$ and $y$ are in separate parts of $P$ is at most $\sigma \cdot
d(x, y) / l$. Each part of the partition $P$ is represented by an
arbitrarily selected point in it, and for each point $x \in \Sigma$,
we use $z(x)$ to denote the representative point of the part of $P$
which contains $x$. Alice and Bob both determine $z$ using shared
random bits.

Define $\wt{x} \in \Sigma^{n_1}$ and $\wt{y} \in \Sigma^{n_2}$ by
$\wt{x}_i = z(x_i)$ and $\wt{y}_i =z(y_i)$ for every $i$. Note that
Alice is able to compute $\wt{x}$ and Bob is able to compute $\wt{y}$
without any communication.

We now describe our protocol for $(r, \alpha)$-gap DTW. Define
$$\alpha' = \alpha c_1 / (82 \sigma).$$ Using
Proposition~\ref{propboundedKDTW}, Alice and Bob run a $0.99$-accurate
protocol for the bounded $\alpha'$-DTW problem on $\wt{x}$ and
$\wt{y}$. Note that the time-efficient version of this protocol uses space
$$O(n \alpha'^{-1} \log \alpha' \cdot \log n) = O(n\alpha^{-1} \sigma \log \alpha \cdot \log n),$$
and that the time-inefficient version of this protocol uses space
$$O(n \alpha'^{-1} \cdot \log n) = O(n\alpha^{-1} \sigma \cdot \log n).$$
If this protocol
outputs ``Fail'', then Bob returns ``1''.  Otherwise the protocol
outputs a number $Z$.  Bob outputs $1$ if $Z\ge nr/(2\alpha')$ and ``0''
otherwise.
	
To show the correctness, it suffices to show that with probability at
least $.9$, if $\dtw(x, y) \le nr / \alpha$, then $\dtw_0(\wt{x},
\wt{y}) \le n / \alpha'$ and $\dtw(\wt{x}, \wt{y}) \le
\frac{nr}{2\alpha'}$, and that if $\dtw(x, y) > nr$, then
$\dtw(\wt{x}, \wt{y}) > nr/2$.

	We complete the proof with case analysis.
	\paragraph{Case 1}
	Suppose that $\dtw(x,y)\le nr/\alpha$. We wish to show that with probability at least $.9$,
        \begin{equation}
          \dtw_0(\wt{x}, \wt{y}) \le n / \alpha',
          \label{eqprop1}
        \end{equation}
        and
        \begin{equation}
          \dtw(\wt{x}, \wt{y}) \le \frac{nr}{2\alpha'}.
          \label{eqprop2}
        \end{equation}

	Let $C^*=(\wb{x}, \wb{y})$ be an optimal DTW
        correspondence between $x$ and $y$ and assume without loss of
        generality that $C^*$ is of length no greater than $2n$. Let
        $E$ be the set indices of edges of $C^*$ being cut by the partition,
        i.e.,
	\[
	E=\{i\in [|\wb{x}|]: \wb{\wt{x}}_i\neq  \wb{\wt{y}}_i\}.
	\]   Let $E_1$ be the set of indices of edges in $C^*$ with length at most $l$ and
	$E_2$ be the set of indices of edges with length greater than $l$, i.e.,
	\[
	E_1=\{i\in [|\wb{x}|]: d(\wb{x}_i,\wb{y}_i)\le l\} \quad\text{and}\quad E_2=\{i\in[|\wb{y}|]: d(\wb{x}_i, \wb{y}_i)> l\}.
	\]
	Because $\dtw(x, y) \le nr / \alpha$,
	\[
	|E_2| \le \frac{nr}{\alpha l} = \frac{n}{c_1\alpha}.
	\]
	For each $i\in E_1$, the corresponding edge it is cut by the grid with probability at most $\sigma \cdot d(\wb{x}_i,\wb{y}_i)/l$. Thus we can compute the expected size of $E$ as
	\[
	\EE[|E|] \le \sum_{i\in E_1} \frac{\sigma \cdot d(\wb{x}_i,\wb{y}_i)}{l} + |E_2| \le \frac{\sigma \cdot \dtw(x,y)}{l} + |E_2| \le  \frac{2 n\cdot \sigma}{c_1\alpha}.
	\]
	By Markov's inequality, with probability at least $0.9$,
	\begin{equation}
	  |E| \le \frac{20n\cdot \sigma}{c_1\alpha},
          \label{eqEboundedsize}
	\end{equation}
	which we condition on for the rest of this case. Notice that
        $\dtw_0(\wt{x}, \wt{y}) \le |E|$, meaning that
        \eqref{eqEboundedsize} allows us to bound $\dtw_0(\wt{x},
        \wt{y})$ by
        $$\dtw_0(\wt{x}, \wt{y}) \le \frac{20n\sigma}{c_1\alpha} < n / \alpha',$$
        which establishes \eqref{eqprop1}, as desired.

        Moreover, for each edge being cut, its length is increased by
        at most $2l$. Thus, we have,
	\[
	  \dtw(\wt{x}, \wt{y}) \le \dtw(x, y) + 2l \cdot \dtw_0(x, y),
       	  \]
          which by \eqref{eqEboundedsize}, implies that
          
          \begin{align*}  
	    \dtw(\wt{x}, \wt{y}) & \le \dtw(x, y) + 2l \cdot \frac{20n\cdot \sigma}{c_1\alpha} \\
                                 & \le nr / \alpha + 2c_1 r \cdot \frac{20n\cdot \sigma}{c_1\alpha} \\
            & \le \frac{41nr\cdot \sigma}{\alpha} \le \frac{nr}{2\alpha'}, \\
          \end{align*}
       	  
            establishing \eqref{eqprop2}, as desired.
	
	\paragraph{Case 2}
	Next, we consider the case that $\dtw(A, B) >  r n$. In this case, we wish to show that
        \begin{equation}
          \dtw(\wt{x}, \wt{y}) > nr/2.
          \label{eqprop3}
        \end{equation}
        Consider an optimal correspondence $\wt{D^*}$ between $\wt{x}$
        and $\wt{y}$ which, without loss of generality, has at most
        $2n$ edges.  Note that $\wt{D^*}$ also defines a
        correspondence $D^*$ between $x$ and $y$. Each edge in
        $\wt{D^*}$ costs at most $2l$ less than the corresponding edge
        in $D^*$. Since $D^*$ has at most $2n$ edges, it follows that
        
        \begin{align*}
          \dtw(\wt{x}, \wt{y}) & \ge \dtw(x, y) - 2n \cdot 2l  \\
                               & \ge rn - 4c_1 nr.
        \end{align*}
        For $c_1$ small enough, we get that $\dtw(\wt{x}, \wt{y}) >
        nr/2$, establishing \eqref{eqprop3}, as desired.
\end{proof}

The proof of Theorem \ref{thmmainthmgeneralized} also require a
slightly strengthened version of Lemma \ref{lemma:gap_to_approx},
provided next.
\begin{lemma}[From Gap DTW to Approx.]
	\label{lemma:gap_to_approx2}
	Let $\delta\in(0,1)$ be an error parameter and suppose
        $\Sigma$ is a metric space with polynomially bounded aspect
        ratio $\kappa$. Suppose for any $r\ge 1$ and $1 \le \alpha\le n$
        there is a $(1-\delta)$-accurate one-way protocol for $(r,
        \alpha)$-gap DTW$(\Sigma^{\le n})$ with space
        complexity $s_1(r, \alpha, \delta)$, and a
        $(1-\delta)$-accurate one-way protocol for bounded
        $\alpha$-DTW($\Sigma^{\le n}$) with space complexity
        $s_2(\alpha, \delta)$, then there is a
        $(1-\delta)$-accurate one-way protocol for 
        $\alpha$-DTW$(\Sigma^{\le n})$ with space
        complexity \[ O\Bigg[\sum_{i=0}^{\lceil\log
            (2n\kappa\alpha)\rceil}s_1\bigg(2^i, \frac{\alpha}{2},
          \frac{\delta}{c \log \log n}\bigg) + s_2\bigg[\alpha,
            \frac{\delta}{2}\bigg]\Bigg],
	\]
        for some constant $c$.

         Moreover, if the protocols for $(r, \alpha)$-gap DTW and
         bounded $\alpha$-DTW are efficient, then so is the protocol
         for $\alpha$-DTW.
\end{lemma}
\begin{proof}
	Let $c$ be a constant whose value is selected to be
        sufficiently large. Without loss of generality the smallest
        distance in $\Sigma$ is $1$ and the largest distance is
        $\kappa$. In order to solve $\alpha$-DTW, Alice runs
        essentially the same protocol as described in the proof of
        Lemma \ref{lemma:gap_to_approx}, although with slightly
        different parameters. In particular, she runs a
        $(1-\delta/2)$-accurate protocol $\wt{\Pi} =(\wt{\sk},
        \wt{F})$ for the bounded $\alpha$-DTW problem and
        $(1-\frac{\delta}{c \log \log n})$-accurate protocols $\Pi_i
          =(\sk_i, F_i)$ for the $(2^i, \alpha / 2)$-gap DTW problem
          on $x$ for each $i \in \{0, 1, \ldots,\lceil \log(2n\kappa\alpha)
          \rceil\}$.  Alice then sends $\wt{\sk}(x), {\sk}_1(x),
                  {\sk}_2(x), \ldots, {\sk}_{\lceil \log(2n\kappa\alpha)
                    \rceil}(x)$ to Bob.

         Bob's decoding procedure, however, is slightly more advanced
         then in the proof of Lemma \ref{lemma:gap_to_approx}. If
         $\wt{F}(\wt{\sk}(x), y)$ is not ``Fail'', then Bob returns
         $\wt{F}(\wt{\sk}(x), y)$ as the answer for $\alpha$-DTW.
         Otherwise, if $F_0(\sk_0(x), y) = 0$, he return $n
         / \alpha$. Finally, if $\wt{F}(\wt{\sk}(x), y)$ is ``Fail''
         and $F_0(\sk_0(x), y) = 1$, then Bob searches for a
         value of $i$ such that $F_{i - 1}(\sk_{i - 1}(x),
         y) = 1$ and $F_i(\sk_i(x), y) = 0$, and then
         returns $2^i n / \alpha$.

         Bob searches for $i$ via a binary search: He first examines
         ${F}_i({\sk_i}(x), y)$ for $i$ roughly half-way between
         $[0, \lceil \log(2n \kappa\alpha) \rceil]$; if
         ${F}_i({\sk_i}(x), y) = 0$, then he recurses on the
         first half of the region, and if ${F}_i({\sk_i}(x), y)
         = 1$, then he recurses on the second half of the
         region. Because ${F}_0({\sk_0}(x), y) = 1$ and
         ${F}_{\lceil \log(2n \kappa\alpha) \rceil}{\sk}_{\lceil
           \log(2n \kappa\alpha) \rceil}(x), y) = 0$ (as shown in the proof
         of Lemma \ref{lemma:gap_to_approx}), the binary search is
         guaranteed to succeed at finding some $i$ for which
         ${F}_{i - 1}({\sk_{i - 1}}(x), y) = 1$ and
         ${F}_i({\sk_i}(x), y) = 0$.

         By the same reasoning in the proof of
         \ref{lemma:gap_to_approx}, if all of the sketches
         $\wt{\sk}(x)$ and ${\sk}_i(x)$ are correct, then Bob will
         correctly solve $\alpha$-DTW. Rather than conditioning on all
         of the sketches being correct, however, we can now condition
         only on the $O(\log \log (n\kappa\alpha)) = O(\log \log n)$ sketches
         which Bob evaluates being correct. By the union bound, for
         $c$ sufficiently large, this occurs with probability at least
         $1 - \delta$, as desired.
	
	The space-complexity of the above protocol follows from a
        direct calculation, completing the proof.
\end{proof}

By combining Proposition \ref{propboundedKDTW}, Lemma
  \ref{lemma:r_a_gap_general} with $\delta$ set to $O(1 / \log \log
  n)$, and Lemma \ref{lemma:gap_to_approx2}, the proof of Theorem
  \ref{thmmainthmgeneralized} is complete.

\end{document}